\documentclass[12pt]{article}

\usepackage{amsfonts, amsmath, amssymb}
\usepackage{a4wide,color}
\usepackage{tikz, url}

\newcommand{\eq}{\leftrightarrow}

\newcommand{\imp}{\rightarrow}

\newcommand{\et}{\wedge}
\newcommand{\vel}{\vee}
\newcommand{\Et}{\bigwedge}
\newcommand{\Vel}{\bigvee}

\newcommand{\is}{\exists}
\newcommand{\T}{\top}

\renewcommand{\phi}{\varphi}

\newcommand{\Union}{\bigcup}
\newcommand{\inter}{\cap}
\newcommand{\Inter}{\bigcap}

\newcommand{\M}{\widehat{K}}

\newcommand{\bisim}{{\raisebox{.3ex}[0mm][0mm]{\ensuremath{\medspace \underline{\! \leftrightarrow\!}\medspace}}}}
\newcommand{\Nat}{\mathbb N}
\newcommand{\Naturals}{\Nat}

\newcommand{\np}{\overline{p}}

\newcommand{\nq}{\overline{q}}

\usepackage[thmmarks]{ntheorem}
\theorembodyfont{\rmfamily}
\theoremsymbol{\ensuremath{\dashv}}
\usepackage{newproof}

\newtheorem{theorem}{Theorem}
\newtheorem{example}[theorem]{Example}

\newtheorem{proposition}[theorem]{Proposition}

\newcommand{\lang}{\mathcal L}
\newcommand{\model}{\mathcal M}
\newcommand{\C}{\mathcal C}
\newcommand{\VV}{\mathcal V}
\newcommand{\FF}{\mathcal F}
\newcommand{\weg}[1]{}
\newcommand{\pre}{\mathsf{pre}}
\newcommand{\post}{\mathsf{post}}

\newcommand{\powerset}{\mathcal{P}}
\newcommand{\RR}{\mathfrak R}

\newcommand{\CD}{\mathit{CD}}
\newcommand{\lbr}{[\![}
\newcommand{\rbr}{]\!]}
\newcommand{\I}[1]{\lbr #1 \rbr} 

\newcommand{\comment}[4][inline]{
  \ifthenelse{\equal{#1}{margin}}
  {
    \marginpar{\scriptsize \hrule\noindent
    {\bf #2:\\} {\em\textcolor{#3}{#4}} \hrule}
  }
  {
    \hrule\noindent {\bf #2:} {\em\textcolor{#3}{#4}} \hrule
  }
}

\newcommand{\jl}{\weg}

\begin{document}

\title{Knowledge and simplicial complexes}
\author{Hans van Ditmarsch\thanks{CNRS, LORIA, Universit\'e de Lorraine, France; {\tt hans.van-ditmarsch@loria.fr}}  \and \'Eric Goubault\thanks{LIX, \'Ecole Polytechnique, CNRS, Institut Polytechnique de Paris, 91191 Palaiseau, France; {\tt goubault@lix.polytechnique.fr}} \and J\'er\'emy Ledent\thanks{University of Strathclyde, Glasgow, Scotland; {\tt jeremy.ledent@strath.ac.uk}}  \and Sergio Rajsbaum\thanks{UNAM, Mexico D.F., Mexico; {\tt sergio.rajsbaum@gmail.com}}}
\date{\today}
\maketitle

\begin{abstract}
Simplicial complexes are a versatile and convenient paradigm on which to build all the tools and techniques of the logic of knowledge, on the assumption that initial epistemic models can be described in a distributed fashion. Thus, we can define: knowledge, belief, bisimulation, the group notions of mutual, distributed and common knowledge, and also dynamics in the shape of simplicial action models. We give a survey on how to interpret all such notions on simplicial complexes, building upon the foundations laid in \cite{goubaultetal:2018}.  
\end{abstract}

\section{Introduction}

Epistemic logic investigates knowledge and belief, and change of knowledge and belief, in multi-agent systems. Since~\cite{hintikka:1962} it has developed into multiple directions. Knowledge change was extensively modelled in temporal epistemic logics~\cite{Alur2002,halpernmoses:1990,Pnueli77} and more recently in dynamic epistemic logics~\cite{hvdetal.del:2007}, where action model logic was particularly successful~\cite{baltagetal:1998}. Modelling asynchrony and concurrency in epistemic logic has been investigated in many works in a temporal or dynamic setting~\cite{dixonetal.handbook:2015,halpernmoses:1990,hareletal:2000,Knight13,panangadenetal:1992,peleg:1987,Pnueli77} and more recently in dynamic epistemic logic~\cite{balbianietal:2019,degremontetal:2011,KnightMS19}.

Combinatorial topology has been used to great effect in distributed computing to model concurrency and asynchrony since the early works~\cite{BiranMZ90,FischerLP85,luoietal:1987} that identified one-dimensional connectivity invariants, and  the higher dimensional topological properties discovered in the early 1990's e.g.~\cite{HerlihyS93,HS99}. A recent overview is~\cite{herlihyetal:2013}. The basic structure in combinatorial topology is the \emph{simplicial complex,} a collection of subsets of a set of vertices closed under containment. The subsets, called
\emph{simplices}, can be vertices, edges, triangles, tetrahedrons, etc.
\begin{center}
\includegraphics[scale=0.08]{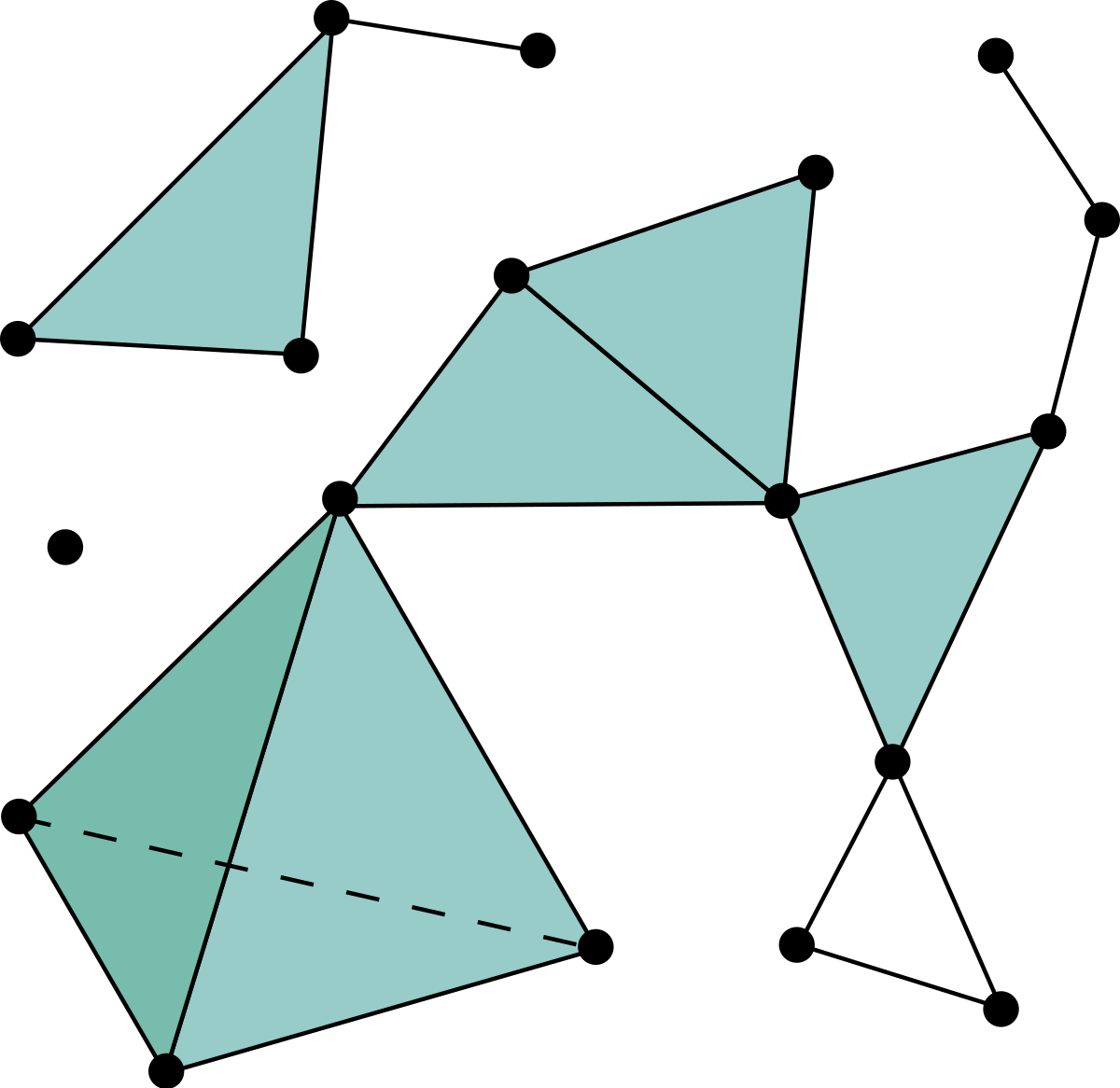}
\label{fig:complexes}
\end{center}
 Geometric manipulations that preserve some of the topology of an initial simplicial complex, such as subdivision, correspond in a natural and transparent way to the evolution of a distributed computation by a set of agents. 
 In turn, the topological invariants preserved, determine the computational power of the distributed computing
 system (defined by the possible failures, asynchrony, and communication media).
 
 Recently, an epistemic logic interpreted on simplicial complexes has been proposed~\cite{goubaultetal:2019,goubaultetal:2018,ledent:2019}. There is a categorical correspondence between the usual Kripke model frames to interpret epistemic logic and models based on simplicial complexes. The epistemic logic of~\cite{goubaultetal:2018} is also a dynamic epistemic logic: the action models of~\cite{baltagetal:1998} correspond to simplicial complexes in the same way as Kripke models, but then with a dynamic interpretation, such that action execution is a product operation on simplicial complexes. In this way,~\cite{goubaultetal:2018} models distributed computing tasks and algorithms in their combinatorial topological incarnations as simplicial complexes. As their knowledge operators can be interpreted for typical asynchronous algorithms (e.g. when processes communicate by writing and reading shared variables), this is a relevant notion of asynchronous knowledge, as it reasons over indistinguishable asynchronous sequences of events, although differently so than \cite{degremontetal:2011}. Predating \cite{goubaultetal:2018} the work of Porter \cite{Porter} already established a relation between multi-agent epistemic models and simplicial complexes by way of the logic of knowledge. In that work the dynamics were based on the runs-and-systems approach of \cite{halpernmoses:1990,faginetal:1995} and not on a dynamic epistemic logic as in~\cite{goubaultetal:2018}. 

In this work we further explore and survey what the logic of knowledge can contribute to the description of simplicial complexes and their dynamic evolution. The goal is to present an introduction to a variety of issues that arise when using simplicial complexes in epistemic logic; each one can be studied in more detail, and many problems remain open.

\paragraph*{Organization}
This is an overview of our contribution. Section~\ref{sec:motivation} introduces informally, and motivates the simplicial complex approach to epistemic logic.
Section~\ref{section.tools} presents the logical and topological tools and techniques that we build upon in this work and also reviews the results of \cite{goubaultetal:2018}. Section~\ref{section.bisimulation} defines bisimulation for simplicial complexes, and gives the obvious adequacy results for the notion. Section~\ref{section.local} presents a local semantics for the logic of knowledge on simplicial complexes. In~\cite{goubaultetal:2018} the designated object for interpreting formulas always is a facet of the simplicial model under consideration. By local semantics we mean that the designated object can be any simplex. Local semantics therefore allow us to check the truth even in a vertex of a complex, that represents that local state of an agent. Section~\ref{section.globallocal} investigates when epistemic models (Kripke models wherein all accessibility relations are equivalence relations) can be transformed into local epistemic models with the same information content. An epistemic model is local if for each variable there is an agent such that the value of that variable is uniform in any of her equivalence classes. Section~\ref{section.group} presents semantics on simplicial complexes for the well-known group epistemic notions of mutual, common, and distributed knowledge. In that section we also explore novel group epistemic notions describing higher-dimensional features of complexes, such as truth on a manifold. Section~\ref{section.belief} defines belief for simplicial complexes. Unlike knowledge of propositions, belief of propositions does not entail that these propositions are true. Beliefs may be false. This is useful to model distributed systems wherein agents may be mistaken about what other agents know or believe, or even about their own local state. Section~\ref{section.actionmodels} defines simplicial action models. Unlike~\cite{goubaultetal:2018}, we define preconditions on vertices (from which preconditions for facets can be derived, and vice versa) and we also model factual change (change of the value of variables). As an example of factual change we then model binary consensus in distributed computing. A short concluding section describes further avenues for exploration.

\section{An informal introduction to epistemic logic on simplicial complexes}\label{sec:motivation}

Epistemic logic, also known as the logic of knowledge, has been interpreted on a wide variety of model classes, including Kripke models where all relations are equivalence relations (a.k.a.\ epistemic models)~\cite{faginetal:1995}, rough sets~\cite{BanerjeeK07}, neighbourhood structures~\cite{chellas:1980,pacuit:2017}, subset space models~\cite{DabrowskiMP96}, (not necessarily discrete) topological spaces~\cite{parikhetal:2007,aybuke.phd:2017}, and the central topic of this paper, simplicial complexes~\cite{goubaultetal:2018,ledent:2019}. Simplicial complexes provide a fascinating new point of view on epistemic logic.
In this section, we give an intuitive explanation of the notion of simplicial model, its relation with the more traditional epistemic models, and the new questions that arise from this new simplicial point of view. Precise definitions of all the concepts discussed here will be given in Section~\ref{section.logicaltools}.

The main object of study in epistemic logic is the modal operator $K_a \varphi$, which stands for the statement ``the agent $a$ knows that the formula $\varphi$ is true''.
The usual Kripke semantics for this operator is based on an \emph{epistemic model}, which consists of a set~$S$ of \emph{global states}, and for each agent, an equivalence relation on~$S$ called the \emph{indistinguishability} relation.
By definition, the formula $K_a \varphi$ is \emph{true} in some global state~$s$ whenever $\phi$ is true in every state~$t$ which is indistinguishable from~$s$ by the agent~$a$.

In the simplicial complex point of view, each global state is represented by a basic geometric shape called a \emph{simplex}.
The dimension of these simplices depends on the total number of agents that we consider: they are \emph{edges} for~$2$ agents, \emph{triangles} for~$3$ agents, tetrahedrons for~$4$ agents, and more generally, $n$-simplices for $n+1$ agents.
The notion of indistinguishability between global states is encoded geometrically by gluing these basic blocks together.
Thus, simplicial models are of a geometric nature.
Indeed, in the field of combinatorial topology, simplicial complexes are the primary representation of a topological space, see for example~\cite{EHbook2010,kozlov}.

\paragraph*{Epistemic models based on local states}
Global states are the main elements of {epistemic models}, which 
have served as the basis of the most widely used semantics for all varieties of modal logic, since the 1960s.
However, in many situations, the focus is on local states. 
This is for example the case in distributed computing, where the central role of topology for multi-agent computation was discovered.
Here, the agents are parallel processes that communicate with each other
via message-passing or shared memory.
The local state of an agent is defined by its local memory,  program counter, etc.
Indeed, we take most of our examples and motivations from this area.

Another motivation for focusing on local states is the relativistic perspective, 
where at any given moment an agent's \emph{local state} exists and is well-defined but not necessarily so the (global) state.
 In fact, there may be more than one state which is compatible
with a set of local states of the agents. Indeed, even talking about absolute real time
is impossible according to relativity theory, so it may not make sense to talk about the state at some real time $t$.
This relativistic approach has been considered since early on in computer science~\cite{lamport:1978}. 

 \paragraph*{Simplicial models}
 We thus move away from defining an epistemic model in terms of global states, to one defined in terms of local states.
Then, a global state~$s$ for~$n+1$ agents can be decomposed as a tuple $(\ell_0, \ldots, \ell_n)$, consisting of one local state for each agent.
In a simplicial model, each of these local states is modeled separately as a \emph{vertex} of a simplicial complex.
With this point of view, a global state corresponds to a set of $n+1$ vertices, that is, an $n$-dimensional \emph{simplex}.
Two global states are indistinguishable by some agent when it has the same local state in both of them. In a simplicial model, this fact is represented by a situation where one vertex (the local state of the agent) belongs to two different simplices (the two global states).

 To give an intuitive understanding of how simplicial complexes are used to model epistemic situations, we present a few examples adapted from~\cite{goubaultetal:2018,herlihyetal:2013,ledent:2019}.
 In all examples, we work with either two or three agents in order to be able to draw low-dimensional pictures.
 The agents are often represented as colours, gray, white, and black.
 The local state of an agent is written inside or besides a vertex of the corresponding colour.
 Global states are represented as edges (for $2$ agents) or triangles (for $3$ agents).

\begin{example}\label{ex:muddy}
For readers familiar with the muddy children puzzle, the picture below represents the initial knowledge for three children $\{ gray, white, black\}$, viewed as a simplicial model.
The local view of a child is written as a vector of three symbols among $\{0,1,\bot\}$, where each symbol corresponds to gray, white, and black, in that order.
The~`$0$' symbol means no mud, `$1$' means muddy, and `$\bot$' means that
the child cannot see her own forehead.
\begin{center}
\includegraphics[scale=0.25]{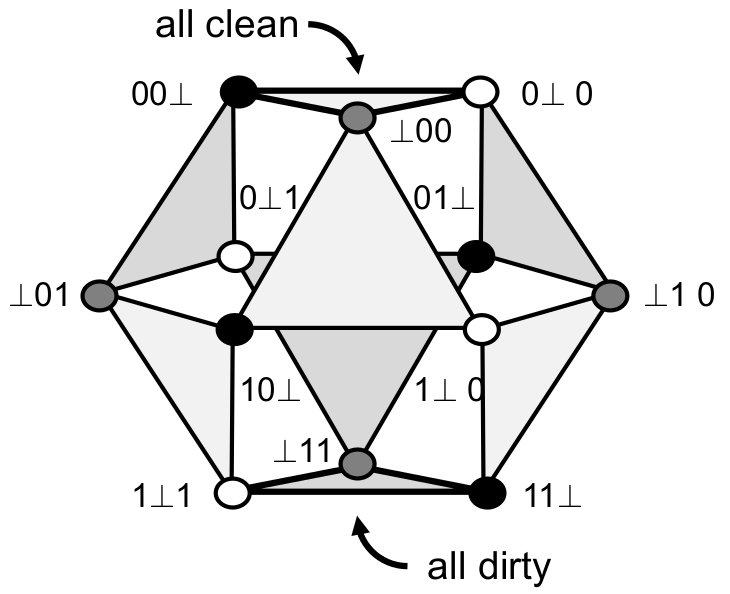}
\label{fig:muddyc}
\end{center}
For instance, the black vertex labeled by $00\bot$ corresponds to the local state of the third child (``$black$''), which is able to see that the other two children are not muddy. Since she cannot see her own forehead, this vertex belongs to two possible global states: one where all children are clean, and one where only the third child is muddy.
On the picture, global states are represented by triangles, and indeed the black vertex labeled by $00\bot$ belongs to two of those triangles.
\end{example}

\begin{example}\label{ex:binInputs}
Each agent gets as input a binary value $0$ or $1$, but does not know which value has been received by the other agents. So, every possible combination of $0$'s and $1$'s is a possible state. In contrast with the previous example, here each agent knows only its own value and not the other agents' values.
The figures below depict this situation for~2 and~3 agents, as epistemic models and as simplicial models.

In the epistemic models, the agents are called $g,w,b$, and the global state is represented as a vector of values, e.g., $101$,
representing the values chosen by the agents $g,w,b$ (in that order).
In the 3-agents case, the labels of the dotted edges have been omitted to avoid
overloading the picture, as well as other edges that can be deduced by transitivity.

In the simplicial model, agents are represented as colours (gray, white and black). The local state is represented as a single value in a vertex, e.g.,
``$1$'' in a gray vertex means that agent $g$ has been given input value~$1$.
The global states correspond to edges in the 2-agents case, and triangles in the 3-agents case.

\medskip
\noindent
\begin{minipage}{0.4\textwidth}
\begin{center}
\begin{tikzpicture}[auto,line/.style={draw,thick,-latex',shorten >=2pt},cloudgrey/.style={draw=black,thick,circle,fill={rgb:black,1;white,3},minimum height=1em},cloud/.style={draw=black,thick,circle,fill=white,minimum height=1em}]
\matrix[column sep=2mm,row sep=2mm]
{
& \node (alpha'') {$01$}; & & \node [cloudgrey] (u) {$0$}; & \node[circle] { }; & \node [cloud] (v) {$1$}; \\
\node (beta'') {$00$}; & & \node (gamma'') {$11$}; & & & \\
& \node (delta'') {$10$}; & & \node [cloud] (w) {$0$}; & & \node [cloudgrey] (z) {$1$}; \\
};
\draw (alpha'') -- node[above left,xshift=2pt,yshift=-2pt] {$\scriptstyle g$} (beta'');
\draw (beta'') -- node[below left,xshift=2pt,yshift=2pt] {$\scriptstyle w$} (delta'');
\draw (gamma'') -- node[xshift=-2pt,yshift=2pt] {$\scriptstyle g$} (delta'');
\draw (alpha'') -- node[xshift=-2pt,yshift=-2pt] {$\scriptstyle w$} (gamma'');
\draw[semithick] (u) -- (v) -- (z) -- (w) -- (u);
\end{tikzpicture}
\end{center}
\end{minipage}
\hfill
\begin{minipage}{0.5\textwidth}
\begin{center}
\begin{tikzpicture}[auto,rotate=45,font=\scriptsize]
\node (a) at (0,0) {$111$};
\node (b) at (1.5,0) {$110$};
\node (c) at (0,1.5) {$011$};
\node (d) at (1.5,1.5) {$010$};
\node (e) at (0.6,0.6) {$101$};
\node (f) at (2.1,0.6) {$100$};
\node (g) at (0.6,2.1) {$001$};
\node (h) at (2.1,2.1) {$000$};
\path[inner sep = 0pt]
      (a) edge node[below right] {$gw$} (b)
          edge node {$wb$} (c)
          edge[dotted] (e) 
      (d) edge node[above left,pos=0.4,xshift=2pt] {$gw$} (c)
          edge node[pos=0.35] {$wb$} (b)
          edge node[right,xshift=1pt,pos=0.3] {$gb$} (h)
      (f) edge node[above right] {$wb$} (h)
          edge node[xshift=1pt] {$gb$} (b)
          edge[dotted] (e) 
      (g) edge node {$gw$} (h)
          edge node[left,xshift=-1pt] {$gb$} (c)
          edge[dotted] (e); 
\end{tikzpicture}
\raisebox{1em}{\includegraphics[scale=0.25]{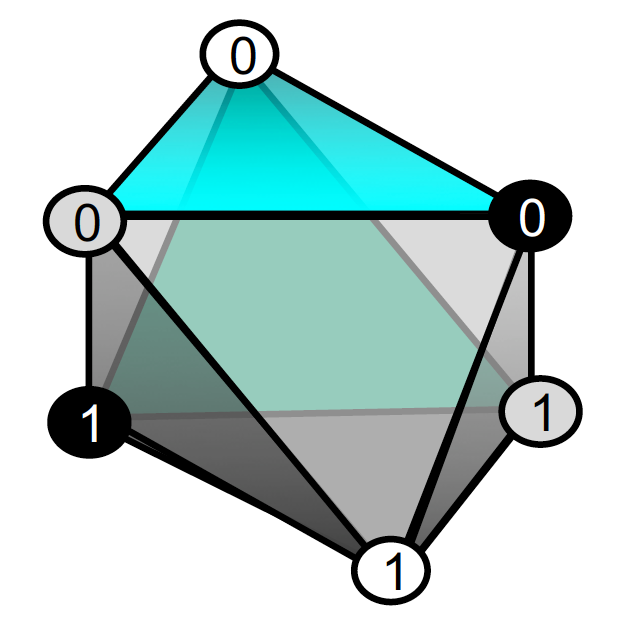}}
\end{center}
\end{minipage}

\medskip
\noindent
Notice that the simplicial complex on the right is an octahedron, that is, a triangulated sphere. More generally, it is known that the binary input simplicial complex for $n+1$ agents is an $n$-dimensional sphere~\cite{herlihyetal:2013}. 
\end{example}

\begin{example}\label{ex:basicDuality}

The toy example below illustrates the equivalence between epistemic models and simplicial models.
The picture shows an epistemic model (left) and its associated simplicial model (right). No variable labelings are shown.
The three agents, named $g, w, b$, are represented as
colours gray, white and black on the vertices of the simplicial complex.
The three global states of the
epistemic model correspond to the three triangles (i.e., $2$-dimensional facets)
of the simplicial complex. The two states indistinguishable by agent~$b$, are glued along their black vertex; while the two states indistinguishable by agents $g$ and $w$ are glued along the gray-and-white edge. 
The two functors $\kappa$ and $\sigma$ translating back and forth between the two kinds of models are defined formally in Section~\ref{sec:epistLoSimModels}.

\begin{center}
\begin{tikzpicture}[auto,dot/.style={draw,circle,fill=black,inner sep=0pt,minimum size=3pt},cloudgrey/.style={draw=black,thick,circle,fill={rgb:black,1;white,2},inner sep=0pt,minimum size=8pt},cloud/.style={draw=black,thick,circle,fill=white,inner sep=0pt,minimum size=8pt}, cloudblack/.style={draw=black,thick,circle,fill=black,inner sep=0pt,minimum size=8pt}]
\node[dot] (p) at (-1,0) {};
\node[dot] (q) at (0,0) {};
\node[dot] (r) at (1,0) {};
\draw (p) -- node[above] {$g,w$} (q);
\draw (q) -- node[above] {$b$} (r);
 
\path[->, bend right, >=stealth] (3.3,0.2)  edge node[above] {$\kappa$} (1.7,0.2);
\path[->, bend right, >=stealth] (1.7,-0.2)  edge node[below] {$\sigma$} (3.3,-0.2);

\draw[thick, draw=black, fill=blue, fill opacity=0.15]
  (4,0) -- (5,-0.577) -- (5,0.577) -- cycle;
\draw[thick, draw=black, fill=blue, fill opacity=0.15]
  (5,-0.577) -- (5,0.577) -- (6,0) -- cycle;
\draw[thick, draw=black, fill=blue, fill opacity=0.15]
  (6,0) -- (7,-0.577) -- (7,0.577) -- cycle;
\node[cloudblack] (b1) at (4,0) {};
\node[cloudgrey] (g1) at (5,-0.577) {};
\node[cloud] (w1) at (5,0.577) {};
\node[cloudblack] (b2) at (6,0) {};
\node[cloudgrey] (g2) at (7,-0.577) {};
\node[cloud] (w2) at (7,0.577) {};
\end{tikzpicture}
\end{center}
\end{example}

\jl{The order of colors in Examples \ref{ex:subdivisions} and \ref{ex:fig-synch} is not coherent with the other examples; I might fix it later}

\begin{example}[Asynchronous computation and subdivisions]
\label{ex:subdivisions}
We now give an example that illustrates how a simplicial model can change after
communication. In Section~\ref{section.actionmodels} we explore change of information using action models.
We describe here an important geometric operation on simplicial complexes called the \emph{chromatic subdivision}.
We describe it for 3 agents, but it generalizes to any number of agents. 
It appears in various distributed computing situations, where agents communicate by
shared memory or message passing~\cite{herlihyetal:2013}.

Consider three agents, represented with colours black, gray, white, and suppose that the black agent may
receive as input value either $0$ or $1$, while the two other agents receive input $0$.
 The two triangles on the left represent two facets of the input model, with input values $000$ (green) and $100$ (yellow). Now, suppose they communicate to each other their inputs, but in the following unreliable way.
After communication, it is possible that an agent missed hearing the input value 
 of one or two of the other agents. There are 4 possible patterns of communication, $x,y,z,w$,
 illustrated in the figure. Each pattern is defined by a table, with one row per agent.
 An agent $x$ receives the value from  agent $y$ if and only if they are both on the same row,
 or $x$ is in a row below that of $y$.
 After communication, we have the simplicial model on the right. Each vertex is labeled with
a vector indicating which values it received.
\begin{center}
\includegraphics[scale=0.5]{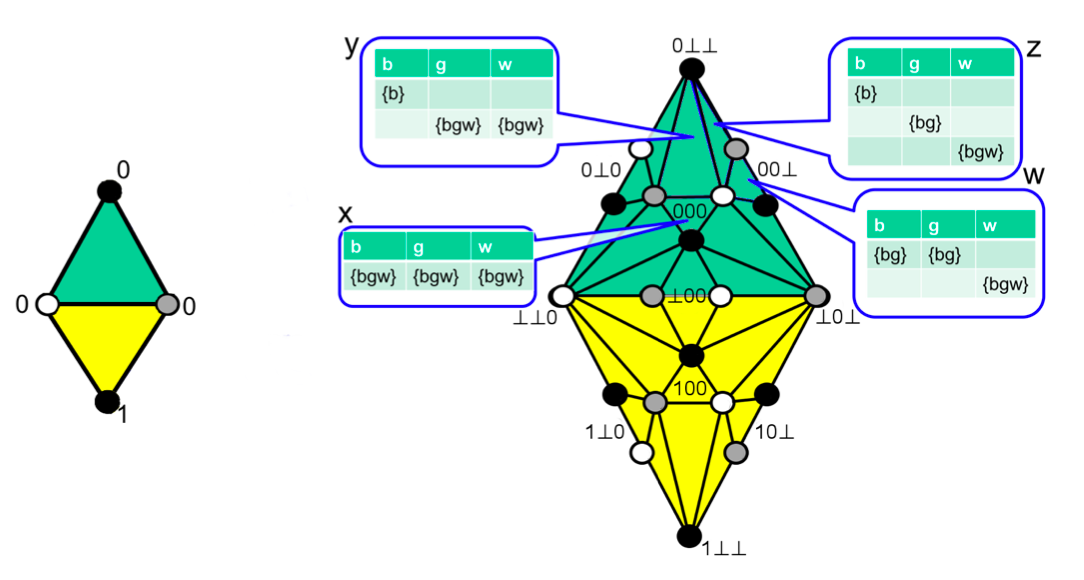} 
\end{center}
The triangle $x$ represents the situation where all agents heard their inputs (all $0$ in this case) from each other,
and intersects in an edge with the triangle $y$, where the black agent did not hear from the other two.
Triangle $x$ intersects only in a vertex with triangle $w$, where both the gray and the black heard from each other, but not from white.
Triangle $x$ also intersects in only one vertex with triangle $z$, where the black heard from no-one, gray heard only from black, and white heard from everybody.
\end{example}

\begin{example}[Synchronous computation]
\label{ex:fig-synch}
Consider again black, white, and gray agents. They get inputs $012$, respectively.
 Suppose they communicate to each other their inputs in a reliable way, but now, at most one of the agents may crash.
Thus, when they communicate, it is possible that an agent missed hearing the input value of another agent,
and in that case, it knows that this agent crashed. 
  After communication, we have the simplicial model below.
  The center triangle represents the state where no one crashed, and each agent learned the inputs of every agent.
 The other edges represent the cases where one agent crashed, and that is why the dimension goes
 down. Thus, such global states are represented as edges instead of triangles.
 The local state of the crashed agent is no longer relevant.
\begin{center}
\includegraphics[scale=0.25]{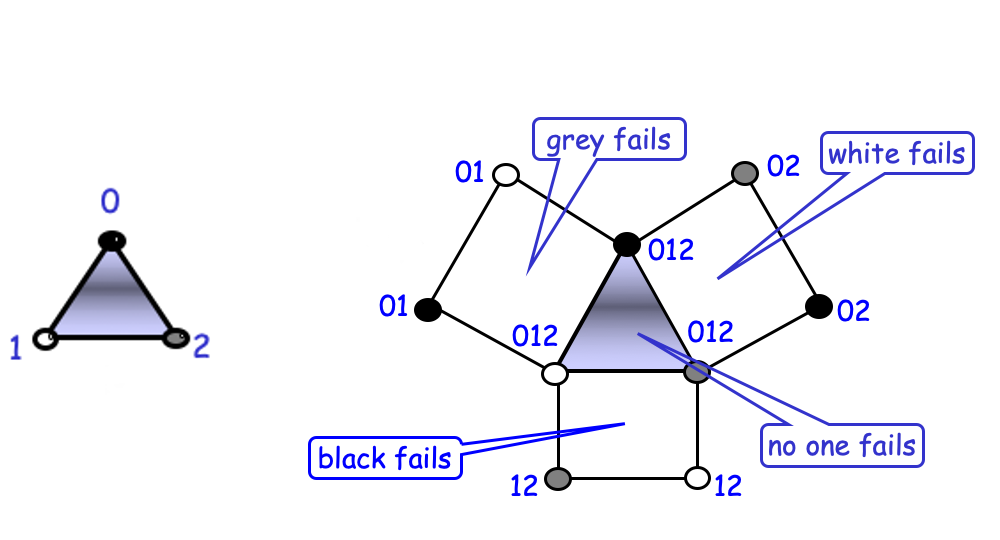}
\end{center}
This is an example of synchronous computation, which has been thoroughly studied in distributed computing,
see e.g.~\cite[chapter 13.5.2]{herlihyetal:2013}.
\end{example}

The correspondence between epistemic models and simplicial models exhibited in~\cite{goubaultetal:2018} (and depicted in Example~\ref{ex:basicDuality}) can be extended to model various epistemic notions.
We conclude this motivating section with an overview of novel ways presented in this contribution in which epistemic logic can be relevant to describe simplicial complexes, with pointers to further sections giving details.
\begin{itemize}
\item
{\bf Locality.} 
A crucial property that allows one to translate an epistemic model into a simplicial one is called \emph{locality} (see Section~\ref{section.logicaltools}). Intuitively, an epistemic model is local when each propositional variable talks only about the local state of one of the agents.
This condition can be lifted in the case of \emph{finite} epistemic models as demonstrated in~\cite{ledent:2019}.
In Section~\ref{section.globallocal} we investigate the consequences of this solution for the information content of models, and what happens in the infinite case.
\item
{\bf Dynamics.}
A simplicial complex version of the \emph{action models} found in dynamic epistemic logic was defined in~\cite{goubaultetal:2018}.
In Section~\ref{section.actionmodels}, we build on this idea in order to include the notion of \emph{factual change}, which allows the values of propositional variables to be modified by the action model.
\item 
{\bf Bisimulation.} A central topic in modal logics is that of bisimulation between epistemic models. Bisimulation of simplicial models has been proposed  in~\cite{GoubaultLLR19,ledent:2019}.
We further discuss simplicial bisimulation in Section~\ref{section.bisimulation}, and relate it with simplicial maps and covering spaces.
\item
{\bf Maps.} 
In topology maps from one space to another are continuous functions, which in the case of simplicial complexes
are \emph{simplicial maps}.
%
The simulations and bisimulations introduced in Section~\ref{section.bisimulation} are natural generalizations of simplicial maps, on condition that they are also value preserving. Then, in Section~\ref{section.belief} we model the notion of belief (knowledge that may be incorrect) on simplicial complexes, by way of so-called {\em belief functions} that induce simplicial maps in an obvious way, however in that case they need not be value preserving, as beliefs may be incorrect.
\item
{\bf Topology.}
The most evident new feature of simplicial models is their geometric nature. Hence, one can analyze the topological structure of the underlying space, such as its high-dimensional connectivity (e.g.\ homotopy groups, representing its ``holes'').
Examples~\ref{ex:muddy} and~\ref{ex:binInputs} make the geometric
space evident that would be hard to see in an epistemic model.

The basic knowledge operator $K_a\phi$ is intrinsically one-dimensional. However,
higher dimensional structure in complexes can be described by, for example, group epistemic operators, such as the common distributed knowledge presented in Section~\ref{section.group} that can check truth on manifolds.
\item
{\bf Local semantics.} 
In epistemic models the typical unit of interpretation is the (global) state, as in the semantics of epistemic logic on simplicial complexes presented in~\cite{goubaultetal:2018,ledent:2019}. However, there are far more relevant units of interpretation on simplicial complexes, as we should not be constrained to determine truth in facets only but should want to determine truth in simplices of any dimension. Already in Example~\ref{ex:basicDuality}, we see that there are 3 states, corresponding to three facets, while the total number of simplices is 17. Local interpretation on simplices is  addressed in Section~\ref{section.local}.
\item
{\bf Hierarchical structure.}
In distributed computing simplicial models often have a hierarchical structure, meaning that they are composed of submodels
representing the behavior of systems with less agents, and hence of lower dimension.
In Example~\ref{ex:subdivisions}, notice that the simplicial model on the right is a subdivision of the initial simplicial model, and
furthermore, the green sub-complex is a subdivision of the initial green triangle, while the yellow sub-complex
is a subdivision of the initial yellow triangle. Notice that  whenever the white and gray agents do not hear from the black
one, they cannot tell what its input was.
Thus, the intersection $\C'$ of both the green and the yellow complexes consists of three edges (and their vertices).
This complex $\C'$  is a subdivision of the initial edge, which is the intersection of the green and the yellow triangle.
Hence, we can analyze a subcomplex  $\C'$ of $\C$, of dimension $1$, with an interesting semantic meaning: it represents 
a subsystem consisting of only two agents, white and gray.
One can ask what the white and gray agents know on vertices or edges  of $\C'$, as if the larger system of three agents did not exist.

This sort of semantics for edges instead of triangles is addressed in full generality in Section 6.1, as language-restricted local semantics.
\item
{\bf Variable number of agents.} 
In distributed computing, whenever crashes are detectable, facets of different size may occur (see Example~\ref{ex:fig-synch}). In epistemic logic such phenomena are modelled in logics of awareness and knowledge, and in particular the kind of unawareness known as {\em unawareness of propositional variables}, as in \cite{faginetal:1988,hvdetal.ieee:2009}. A thorough treatment of this topic is deferred to future research. Another direction for future investigation are adjunctions between epistemic frames and simplicial complexes, that generalize the (equivalent) correspondence between frames and complexes of \cite{goubaultetal:2018}.
\end{itemize}
For the sake of completeness, we also mention below some of the important features of simplicial models, even though we do not address them in the rest of the paper.
They are not specific to the epistemic logic perspective, and have been thoroughly studied in distributed computing. See~\cite{herlihyetal:2013} for more detail.
\begin{itemize}
\item
{\bf Topological invariants.}
As it turns out, there are topological
invariants that are maintained starting from an initial simplicial model,
after the agents communicate with each other.
In the most basic case, illustrated in Example~\ref{ex:subdivisions}, the initial simplical model is subdivided, but in other situations where the communication is more reliable ``holes'' may be
introduced, as in  Example~\ref{ex:fig-synch}. Many other examples are described in~\cite{herlihyetal:2013}, and the notion of ``holes'' is formalized in terms of homotopy groups.
\item
{\bf Power of a computational model.}
The topological invariants preserved by the distributed model (number of failures, communication media, etc) in turn
determine what the agents can distributively compute.
After all, whatever computation an agent is performing, the outcome is a function of
its local state. Thus, decisions induce a simplicial map from the  simplicial complex representing what the agents know about the inputs
after communication, to an output simplicial complex model representing what they are expected to compute based on that knowledge. 
This approach is developed in~\cite{goubaultetal:2018}.

\item
{\bf Specifications.}
In~\cite{goubaultetal:2018} it is explained how a simplicial model is used as knowledge that is supposed to be gained by a set of agents
communicating with each other, via a simplicial map from the simplicial model at the end of the agent's communications.
\end{itemize}

\section{Logical and topological tools} \label{section.tools}

\subsection{Logical tools} \label{section.logicaltools}

Given are a set $A$ of $n+1$ \emph{agents} $a_0,\dots,a_n$ (or $a,b,\dots$) and a countable set $P$ of \emph{(global) variables} (or \emph{variables}) $p, q, p', q', \dots$ (possibly indexed). 

\paragraph*{Language}

The {\em language of epistemic logic} $\lang_{K}(A,P)$ is defined as \[ \phi ::= p \mid \neg\phi \mid \phi\et\phi \mid K_a \phi \] where $p \in P$ and $a \in A$. We will write $\lang_{K}(P)$ if $A$ is clear from the context, and $\lang_{K}$ if $A$ and $P$ are clear from the context. For $\neg p$ we may write $\np$. Expression $K_a \phi$ stands for `agent $a$ knows (that) $\phi$.' The fragment without inductive construct $K_a\phi$ is the {\em language of propositional logic} (also known as the Booleans) denoted $\lang_\emptyset(A,P)$. For $\neg K_a \neg \phi$ we may write $\M_a \phi$, for `agent $a$ considers it possible that $\phi$.'

\paragraph*{Epistemic models}

\emph{Epistemic frames} are pairs $\model = (S,\sim)$ where $S$ is the domain of \emph{(global) states}, and $\sim$ is a function from the set of agents $A$ to (binary) \emph{accessibility relations} on $S$, that are required to be equivalence relations. For $\sim\!(a)$ we write $\sim_a$. \emph{Epistemic models} are triples $\model = (S,\sim,L)$, where $(S,\sim)$ is an epistemic frame and where \emph{valuation} $L$ is a function from $S$ to $\powerset(P)$. A pair $(\model,s)$ where $s \in S$ is a \emph{pointed} epistemic model. We also distinguish the \emph{multi-pointed} epistemic model $(\model,S')$ where $S' \subseteq S$. For $(s,t) \in {\sim_a}$ we write $s \sim_a t$, and for $\{ t \mid s \sim_a t \}$ we write $[s]_a$: this is an equivalence class of the relation $\sim_a$. 

We now introduce terminology for epistemic models for distributed systems.

Given an epistemic model $\model = (S,\sim,L)$, a variable $p \in P$ is \emph{local for agent $a \in A$} iff for all $s,t \in S$, if $s \sim_a t$ then $p \in L(s)$ iff $p \in L(t)$. Variable $p$ is \emph{local} iff it is local for some agent $a$, and a model $\model$ is local iff all variables are local. We write $P_a$ for the set of local variables for agent $a$, i.e., $P_a = \{ p \in P \mid \text{for all } s,t \in S, \text{ if } s \sim_a t, \text{ then } p \in L(s) \text{ iff } p \in L(t) \}$. If an epistemic model is local, then $\Union P_a = P$. Without loss of generality we may assume that all members of $\{P_{a_0},\dots,P_{a_n}\}$ are mutually disjoint, and that members of each $P_a$ are subindexed with the agent name as in: $p_a, p'_a, q_a, \dots$

We can also approach locality from the other direction. Let mutually disjoint sets $S_{a_0} \dots S_{a_n}$ of \emph{local states} and mutually disjoint sets $P_{a_0},\dots,P_{a_n}$ of \emph{local variables} (for resp.\ agents $a_0$, \dots, $a_n$) be given. Let for each $a \in A$, $L_a: S_a \imp \powerset(P_a)$ be a valuation for the local states of agent $a$. A \emph{distributed model} is an epistemic model $(S,\sim,L)$ such that:
\begin{itemize}
\item $S \subseteq S_{a_0} \times \dots \times S_{a_n}$;
\item for all $a \in A$ and $s,t \in S$: $s \sim_a t$ iff $s_a = t_a$; 
\item for all $a \in A$ and $s \in S$: $L(s) = \{ p \in P \mid \is a \in A: p \in L_a(s_a) \}$.
\end{itemize}
In other words, each global state $s \in S$ is an $(n+1)$-tuple $s = (s_0,\dots,s_n)$ of local states for the agents $0,\dots,n$ respectively. 

An epistemic frame $(S,\sim)$ is \emph{proper} if $\inter_{a \in A} \sim_a = Id$, where the identity relation $Id : = \{(s,s) \mid s \in S\}$. An epistemic model is proper if its underlying frame is proper. 

Clearly, a distributed model is local proper epistemic model, and we can also see a local proper epistemic model as a distributed model: if $\model = (S,\sim,L)$ is proper, then the model wherein domain $S$ is replaced by the domain $\{ ([s]_{a_0}, \dots, [s]_{a_n}) \mid s \in S\}$ (and wherein the valuation can be relativized to $a$ in the obvious way) is a distributed model. The definition is adequate, as for all $s \in S$ it holds that $[s]_{a_0} \inter \dots \inter [s]_{a_n} = \{s\}$.

An epistemic model is \emph{factual} if all global states have different valuations, i.e., if for all $s,t \in S$ there is a $p \in P$ such that $p \in L(s)$ and $p \notin L(t)$, or $p \notin L(s)$ and $p \in L(t)$.

\paragraph*{Semantics on epistemic models}

The interpretation of a formula $\phi\in \lang_K$ in a global state of a given pointed model $(\model,s)$ is by induction on the structure of $\phi$. The expression ``$\model,s\models \phi$'' stands for ``in global state $s$ of epistemic model $\model$ it holds that $\phi$.''

\[ \begin{array}{lcl}
\model,s \models p & \text{iff} & p \in L(s) \\
\model,s \models \neg\phi & \text{iff} & \model,s \not\models \phi \\
\model,s \models \phi\et\psi & \text{iff} & \model,s \models \phi \text{ and } \model,s \models \psi \\
\model,s \models K_a \phi & \text{iff} & \model,t \models \phi \text{ for all } t \text{ with } s \sim_a t
\end{array} \]
Formula $\phi$ is {\em valid} iff for all $(\model,s)$, $\model,s \models \phi$. Given $(\model,s)$ and $(\model',s')$, by $(\model,s) \equiv (\model',s')$, for 	``$(\model,s)$ and $(\model',s')$ are {\em modally equivalent},'' we mean that for all $\phi \in \lang_K$: $\model,s \models \phi$ iff $\model',s' \models \phi$. Informally, this means that these pointed models contain the same information. We let $\I{\phi}_\model$ stand for $\{ s \in S \mid \model,s \models \phi \}$. This set is called the {\em denotation} of $\phi$ in $\model$.

\medskip

Standard references for epistemic logic are \cite{meyeretal:1995,faginetal:1995,hvdetal.handbook:2015}. The typical example of proper factual local epistemic models are interpreted systems \cite{faginetal:1995}.

\subsection{Topological tools} \label{section.topologicaltools}

We continue by introducing simplicial complexes and other topological tools. A reference on algebraic topology is~\cite{Hatcher}, on a more combinatorial approach there are~\cite{EHbook2010,kozlov}, and on the
use of combinatorial topology in distributed computing it is~\cite{herlihyetal:2013}. Subsequently and in the next subsection we review results from \cite{goubaultetal:2018}, that are also reported in \cite{ledent:2019}. 

\paragraph*{Simplicial complexes}

Given a set of \emph{vertices} $V$ (sometimes also called {\em local states}; singular form {\em vertex}), a \emph{(simplicial) complex} $C$ is a set of non-empty finite subsets of $V$, called \emph{simplices} (singular form {\em simplex}), that is closed under subsets (such that for all $X \in C$, $Y \subseteq X$ implies $Y \in C$), and that contains all singleton subsets of $V$. If $Y \subseteq X$ we say that $Y$ is a \emph{face} of $X$. A maximal simplex in $C$ is a \emph{facet}. The facets of a complex $C$ are denoted as $\FF(C)$, and the vertices of a complex $C$ are denoted as $\VV(C)$. The dimension of a simplex $X$ is $|X|-1$, e.g., vertices are of dimension $0$, while edges are of dimension $1$. The dimension of a complex is the maximal dimension of its facets. A simplicial complex is \emph{pure} if all facets have the same dimension. A {\em manifold} is a pure simplicial complex $C$ of dimension $n$ such that: (i) for any $X,Y \in \FF(C)$ there are facets $X = X_0, \dots, X_m = Y$ such that for all $i < m$ the dimension of $X_i \inter X_{i+i}$ is $n-1$ (this is a `path' from $X$ to $Y$ in $C$), and (ii) any simplex $Z \in C$ of dimension $n-1$ is a face of one or two facets of $C$ (those that are faces of only one facet form the {\em boundary} of $C$).

Complex $D$ is a \emph{subcomplex} of complex $C$ if $D \subseteq C$. A subcomplex of interest is the \emph{$m$-skeleton} $D$ of a pure $n$-dimensional complex $C$, i.e, the maximal subcomplex $D$ of $C$ of dimension $m$.

We decorate the vertices of simplicial complexes with agent's names, that we refer to as  \emph{colours}. 
A \emph{chromatic map} from the vertices of $\C$ to $A$ assigning colours to vertices is denoted $\chi$. 
Thus, $\chi(v)\in A$ denotes that the local state $v$ belongs to agent $a$.
 Different vertices of the same simplex should be assigned different colours. 
 The vertex of a simplex $X$ coloured with $a$ is denoted $X_a$. 
 A pair consisting of a simplicial complex and a colouring map $\chi$ is a \emph{chromatic simplicial complex}. 
 
  From now on, all simplicial complexes will be chromatic pure simplicial complexes.

\paragraph*{Simplicial models}
We now also decorate the vertices of simplicial complexes  with \emph{local variables} $p_a \in P_a$ for $a \in A$, where we recall that $\Union_{a \in A} P_a = P$ (all sets of local variables are assumed to be disjoint). 
While chromatic maps are denoted $\chi,\chi',\dots$,
 \emph{valuations} (valuation maps) assigning sets of local variables to vertices are denoted $\ell, \ell', \dots$
 For any $X \in C$, $\ell(X)$ stands for $\Union_{v \in X} \ell(v)$.
 A \emph{simplicial model} $\C$ is a triple $(C,\chi,\ell)$ where $C$ is a simplicial complex.

\paragraph*{Simplicial maps}
Given simplicial complexes $C$ and $C'$, a \emph{simplicial map} ({\em simplicial function}) is a function $f: \VV(C) \imp \VV(C')$ that preserves simplices (such that for all $X \in C$, $f(X) := \{ f(v) \mid v \in X \} \in C'$). We also let $f(C) := \{ f(X) \mid X \in C \}$. A simplicial map is {\em rigid} if it is dimension preserving (if for all $X \in C$, $|f(X)| = |X|$). A \emph{chromatic simplicial map} is a colour preserving simplicial map, where {\em colour preserving} means that for all $v \in \VV(C)$, $\chi(f(v)) = \chi(v)$. We note that it is therefore also rigid.

Let now $\C = (C,\chi,\ell)$ and $\C' = (C',\chi',\ell')$ be simplicial models. A chromatic simplicial map $f$ between the complexes $C$ and $C'$ is also called a simplicial map between the models $\C$ and $\C'$. The simplicial map $f$ is {\em value preserving} if for all $v \in \VV(C)$, $\ell(f(v)) = \ell(v)$. If $f$ is not only colour preserving but also value preserving, and its inverse $f^{-1}$ as well, then $\C$ and $\C'$ are \emph{isomorphic}, notation ${\C \simeq \C'}$.  
It is customary to define isomorphy between simplicial complexes instead of between simplicial models. However, we will later define a \emph{bisimulation} as a relation between simplicial models, such that isomorphy implies bisimilarity but not necessarily vice versa. We therefore defined isomorphy between simplicial models. 


\subsection{Epistemic logic on simplicial models}\label{sec:epistLoSimModels}

\paragraph*{Semantics on simplicial models}
The interpretation of a formula $\phi\in \lang_K(A,P)$ in a facet $X \in \FF(C)$ of a given simplicial model $\C = (C,\chi, \ell)$ is by induction on the structure of $\phi$. Merely the clause for knowledge looks different than that on epistemic models.
\[ \begin{array}{lcl}
\C,X \models p_a & \text{iff} & p_a \in \ell(X) \\
\C,X \models \neg\phi & \text{iff} & \C,X \not\models \phi \\
\C,X \models \phi\et\psi & \text{iff} & \C,X \models \phi \text{ and } \C,X \models \psi \\
\C,X \models K_a \phi & \text{iff} & \C,Y \models \phi \text{ for all } Y \in \FF(C) \text{ with } a \in \chi(X \inter Y)
\end{array} \]
Validity and modal equivalence are also defined similarly. Formula $\phi$ is {\em valid} iff for all $(\C,X)$, $\C,X \models \phi$; and given $(\C,X)$ and $(\C',X')$, by $(\C,X) \equiv (\C',X')$ we mean that for all $\phi \in \lang_K(A,P)$: $\C,X \models \phi$ iff $\C',X' \models \phi$. 

\paragraph*{Correspondence between simplicial models and epistemic models}
Given agents $A$ and variables $P$, let $\mathcal K$ be the class of local proper epistemic models and let $\mathcal S$ be the class of simplicial models. In \cite{goubaultetal:2018}, simplicial models are shown to correspond to local proper epistemic models by showing that the underlying proper epistemic frames and simplicial complexes correspond (categorically) via functions (functors) $\sigma: \mathcal K \imp \mathcal S$ ($\sigma$ for \emph{S}implicial) and $\kappa: \mathcal S \imp \mathcal K$ ($\kappa$ for \emph{K}ripke), and then extending $\sigma$ and $\kappa$ in the obvious way, such that they are also valuation preserving, to models. The correspondence uses the observed relation between local proper epistemic models and distributed models. These are the details.

Given a local proper epistemic model $\model = (S,\sim,L)$, we define $\sigma(\model) = (C,\chi,\ell)$ as follows. Its vertices are $\VV(\sigma(\model)) = \{ [s]_a \mid a \in A, s \in S\}$, where $\chi([s]_a)=a$. Its facets  are $\FF(\sigma(\model)) = \{ \{[s]_{a_0}, \dots, [s]_{a_n}\} \mid s \in S\}$. For such a facet $\{[s]_{a_0}, \dots, [s]_{a_n}\}$ we write $\sigma(s)$. Then, $[s]_a \in \sigma(s) \inter \sigma(t)$ iff $s \sim_a t$ (i.e., such that $[s]_a = [t]_a$), and for all $p_a \in P$, $p_a \in \ell([s]_a)$ iff $p_a \in L(s)$. We note that $\sigma(\model)$ has dimension $n$.

Given a simplicial model $\C = (C,\chi,\ell)$, we define  an epistemic $\kappa(\model) = (S,\sim,L)$ as follows. Its domain $S$ consists of states $\kappa(X)$ for all $X \in \FF(C)$. We define $\kappa(X) \sim_a \kappa(Y)$ iff $a \in X \inter Y$, and $p_a \in L(\kappa(X))$ iff $p_a \in \ell(X)$. We note that $\kappa(\model)$ is local and proper.

One can now show that:
\begin{proposition}[{\cite{goubaultetal:2018}}] 
\label{prop.corr}
For all $\model \in \mathcal K$, $\kappa(\sigma(\model)) \simeq \model$, and for all $\C \in \mathcal S$, $\sigma(\kappa(\C)) \simeq \C$.
\end{proposition}
\begin{proposition}[{\cite{goubaultetal:2018}}] 
\label{prop.corr2}
For all $\phi \in \lang_K$, $\model,s\models \phi$ iff $\sigma(\model),\sigma(s) \models \phi$, and $\C,X \models\phi$ iff $\kappa(\C),\kappa(X) \models \phi$.
\end{proposition}

\section{Bisimulation for simplicial complexes} \label{section.bisimulation}

If two models are isomorphic, they make the same formulas true, i.e., they contain the same information (what we also called {\em modally equivalent}). In modal logic and concurrency theory, a weaker form of sameness than isomorphy already guarantees that,  \emph{bisimilarity} \cite{baieretal:2008,blackburnetal:2001,sangiorgi:2010,openMinds,glabbeek:1990}.  
As truth is invariant under bisimulation, it is an important way to determine whether structures contain the same information. In distributed computing by way of combinatorial topology, bisimulation between simplicial models has been proposed in~\cite{goubaultetal:2019,ledent:2019}. We present and further explore this notion, and demonstrate its adequacy with respect to bisimulation for epistemic models. This may be considered of interest, for example, if different subdivisions of complexes are shown to be bisimilar, then they contain the same information after all. 

We first define the standard notion of bisimulation between epistemic models, and then bisimulation between simplicial models. 

\paragraph*{Bisimulation between epistemic models}

A \emph{bisimulation} between $\model = (S,\sim,L)$ and $\model' = (S',\sim',L')$, notation $\RR: \model\bisim \model'$, is a non-empty relation $\RR \subseteq S \times S'$ such that for all $s \in S$, $s' \in S'$ with $\RR ss'$ the following three conditions are satisfied:

\begin{itemize}
\item {\bf atoms}: for all $p \in P$, $p \in L(s)$ iff $p \in L'(s')$.
\item {\bf forth}: for all $a \in A$, for all $t$ with $s \sim_a t$, there is a $t'$ with $s'\sim'_at'$ such that $\RR tt'$.
\item {\bf back}: for all $a \in A$, for all $t'$ with $s'\sim'_at'$, there is a $t$ with $s \sim_a t$ such that $\RR tt'$.
\end{itemize}
A bisimulation $\RR$ such that for all $s \in S$ there is a $s' \in S'$ such that $\RR ss'$ and for all $s' \in S'$ there is a $s \in S$ such that $\RR ss'$ is a \emph{total bisimulation}. If there is a bisimulation $\RR$ between $\model$ and $\model'$ we say that $\model$ and $\model'$ are \emph{bisimilar}, denoted $\model \bisim \model'$. A bisimulation between pointed models $(\model,s)$ and $(\model',s')$ is a bisimulation $\RR$ such that $\RR ss'$, notation $\RR: (\model,s) \bisim (\model',s')$, and if there is such a bisimulation we write $(\model,s) \bisim (\model',s')$. A model is \emph{bisimulation minimal} if the only total bisimulation on that model (i.e., between the model and itself) is the identity relation.

\paragraph*{Bisimulation between simplicial models}
Let be given simplicial models $\C = (C,\chi,\ell)$ and $\C' = (C',\chi',\ell')$. A non-empty relation $\RR$ between $\FF(C)$ and $\FF(C')$ is a \emph{(simplicial) bisimulation} between $\C$ and $\C'$, notation $\RR: \C \bisim \C'$, iff for all $Y \in \FF(C)$ and $Y' \in \FF(C')$ with $\RR YY'$ the following three conditions are satisfied:
\begin{itemize}
\item {\bf atoms}: for all $a \in A$ and $p_a \in P_a$, $p_a \in \ell(Y)$ iff $p_a \in\ell(Y')$. 
\item {\bf forth}: for all $a \in A$, if $Z \in \FF(C)$ and $a \in \chi(Y \inter Z)$ there is a $Z' \in \FF(C')$ with $a \in \chi(Y' \inter Z')$ such that $\RR ZZ'$.
\item {\bf back}: for all $a \in A$, if $Z' \in \FF(C')$ and $a \in \chi(Y' \inter Z')$ there is a $Z \in \FF(C)$ with $a \in \chi(Y \inter Z)$ such that $\RR ZZ'$.
\end{itemize}
A \emph{total simplicial bisimulation} $\RR$ is a simplicial bisimulation such that for all $X \in \FF(C)$ there is a $X' \in \VV(C')$ with $\RR XX'$ and for all $X' \in \FF(C')$ there is a $X \in \FF(C)$ with $\RR XX'$. If there is a bisimulation between $\C$ and $\C'$ we write $\C \bisim \C'$. A bisimulation between pointed simplicial models $(\C,X)$ and $(\C',X')$, where $X \in \FF(C)$ and $X' \in \FF(C')$, is a bisimulation $\RR$ such that $\RR XX'$, notation $\RR: (\C,X) \bisim (\C',X')$, and if there is such a bisimulation we write $(\C,X) \bisim (\C',X')$. Relation $\RR$ is a \emph{simulation} if it satisfies the {\bf atoms} and the {\bf forth} conditions. 

\medskip

Intuitively, the {\bf forth}-clause preserves ignorance when going from $\C$ to $\C'$, and the {\bf back}-clause preserves knowledge when going from $\C$ to $\C'$. 

A relation $\RR: \FF(C) \imp \FF(C')$ between facets {\em induces} a (similarly denoted) colour preserving relation $\RR: \VV(C) \imp \VV(C')$ between vertices by way of: if ${\RR}XX'$ then for all $a \in A$, ${\RR}X_aX'_a$. Dually, given a relation $\RR: \VV(C) \imp \VV(C')$ between vertices that is {\em chromatic}, i.e., colour preserving (if $\RR{vv'}$ then $\chi(v)=\chi(v')$), then for any simplices $X,X'$ with $\chi(X)=\chi(X')$, ${\RR}XX'$ denotes that for all $a \in \chi(X)$, ${\RR}X_aX'_a$.

As bisimulations between simplicial models are relations between facets, and as facets correspond to global states in epistemic models, such bisimulations are very much like bisimulations between epistemic models. 

The reader may wonder why bisimulations between simplicial models are not defined between vertices instead of between facets, and with the {\bf back} and {\bf forth} conditions only required between vertices. This is undesirable, because if we were to do so, then  different simplicial models with the same $1$-dimensional skeleton could incorrectly become bisimilar. Example~\ref{example.ins} illustrates how this can go wrong. Similar counterexamples can easily be found if we merely require {\bf back} and {\bf forth} for simplices of dimension $m$ smaller than the dimension $n$ of the complex. This demonstrates that {\bf back} and {\bf forth} need to be required for facets. 

More interestingly, now suppose bisimulations between simplicial models were defined as relations between vertices but with the {\bf back} and {\bf forth} conditions between facets. This is also undesirable: although a relation between facets induces a (unique) relation between vertices, different relations between facets may induce the same relation between vertices. Therefore, only specifying the relation between vertices does not determine a relation between facets. However, there is a unique {\em maximal} such relation, defined as: for all facets $X$ and $X'$, ${\RR}XX'$ iff for all $a \in A$, ${\RR}X_aX'_a$. Example~\ref{example.facevert} illustrates this. 

Given a bisimulation, its induced chromatic relation between vertices interestingly compares to chromatic simplicial maps. 
We recall the notion of chromatic simplicial map as a colour and simplex preserving function between the vertices of a complex. Given $\C = (C,\chi,\ell)$ and $\C' = (C',\chi',\ell')$, we will call a relation $\RR$ between vertices {\em simplex preserving} if for all $X \in C$ with $X$ in the domain of $\RR$ there is a $X' \in C'$ such that ${\RR}XX'$. It is easy to see that a bisimulation $\RR$ between simplicial models induces a relation $\RR$ between vertices that is simplex preserving and such that its converse relation $\RR^{-1} := \{ (v',v) \mid (v,v') \in \RR\}$ is also simplex preserving: 

Let $X \in C$ be such that $X$ is in the domain of $\RR$, then, as $C$ is pure, there is a $Y \in \FF(C)$ with $X \subseteq Y$. Let now $v \in X$. As $v$ is in the domain of $\RR$,  there is a $Z \in \FF(C)$ with $v \in Z$ and a $Z' \in \FF(C')$ such that ${\RR}ZZ'$. As $v \in Z$ and $v \in Y$, $v \in Z \inter Y$. Let now $\chi(v) = a$, then from ${\RR}ZZ'$, $v \in Z \inter Y$ and {\bf forth} it follows that there is a $Y'$ such that ${\RR}YY'$. The set $X' \subseteq \VV(C')$ with ${\RR}XX'$ is therefore a face of $Y'$ and thus a simplex. Similarly, using that {\bf back} is satisfied for $\RR$ we show that $\RR^{-1}$ is simplex preserving. It seems interesting to explore how a bisimulation can be seen as a generalization of a chromatic simplicial map. Note that the {\bf atoms} requirement corresponds to the value preservation of the simplicial map. In Propositions~\ref{prop:sim} and \ref{prop:coveringMain}, later, we also address the  relation between bisimulations (and simulations) and simplicial maps. 

\begin{example} \label{example.bigtri}
The two simplicial models on the left are (totally) bisimilar: all four facets on the far left are related to the single facet on the near left. The two simplicial models on the right are also (totally) bisimilar: the two opposed $\{a0,b0,c0\}$ facets on the near right are related to the facet $\{a0,b0,c0\}$ of the complex on the far right, and the $\{a1,b0,c0\}$ facet on the near right is related to the facet $\{a1,b0,c0\}$ on the far right.

\bigskip

\noindent
\scalebox{.7}{
\begin{tikzpicture}
\node (a0) at (0,0) {$a0$};
\node (b0) at (4,0) {$b0$};
\node (c0) at (2,3.42) {$c0$};
\node (a1) at (3,1.71) {$a0$};
\node (b1) at (1,1.71) {$b0$};
\node (c1) at (2,0) {$c0$};
\draw[-] (a0) -- (c1);
\draw[-] (a0) -- (b1);
\draw[-] (a1) -- (c0);
\draw[-] (a1) -- (c1);
\draw[-] (a1) -- (b0);
\draw[-] (b0) -- (c1);
\draw[-] (b1) -- (c1);
\draw[-] (b1) -- (c0);
\draw[-] (b1) -- (a1);
\node (bisim) at (4.8,1.5) {\Large $\bisim$};
\end{tikzpicture}
\quad 
\begin{tikzpicture}
\node (a0) at (0,0) {$b0$};
\node (b0) at (2,0) {$a0$};
\node (c0) at (1,1.71) {$c0$};
\node (cb) at (1,-1.71) {\color{white}$c0$};
\draw[-] (a0) -- (c0);
\draw[-] (a0) -- (b0);
\draw[-] (b0) -- (c0);
\end{tikzpicture}
\quad\quad\quad
\quad
\quad
\begin{tikzpicture}
\node (a0) at (0,0) {$a0$};
\node (b0) at (4,0) {$b0$};
\node (c0) at (1,-1.71) {$b0$};
\node (a1) at (3,1.71) {$a1$};
\node (b1) at (1,1.71) {$b0$};
\node (c1) at (2,0) {$c0$};
\node (bisim) at (5.2,-.2) {\Large $\bisim$};
\draw[-] (a0) -- (c1);
\draw[-] (a0) -- (b1);
\draw[-] (a1) -- (c1);
\draw[-] (a1) -- (b0);
\draw[-] (b0) -- (c1);
\draw[-] (b1) -- (c1);
\draw[-] (c0) -- (a0);
\draw[-] (c0) -- (c1);
\end{tikzpicture}
\ \ 
\begin{tikzpicture}
\node (a0) at (0,0) {$a0$};
\node (b0) at (4,0) {$b0$};
\node (c0) at (1,-1.71) {\color{white}$b0$};
\node (a1) at (3,1.71) {$a1$};
\node (b1) at (1,1.71) {$b0$};
\node (c1) at (2,0) {$c0$};
\draw[-] (a0) -- (c1);
\draw[-] (a0) -- (b1);
\draw[-] (a1) -- (c1);
\draw[-] (a1) -- (b0);
\draw[-] (b0) -- (c1);
\draw[-] (b1) -- (c1);
\end{tikzpicture}
}
\end{example}

\begin{example} \label{example.ins}
It is not sufficient only to require {\bf forth} and {\bf back}  for lower dimensional simplices, e.g., for edges only. For a counterexample, consider the following simplicial models $\C$ and $\C'$ (and with valuations as depicted). Simplicial model $\C$ consists of four facets/triangles, whereas simplicial model $\C'$ consists of three facets (the middle triangle $\{a1,b1,c1\}$ is not a facet). The 1-skeletons of $\C$ and $\C'$ (checking {\bf forth} and {\bf back} only for edges of triangles) are bisimilar by way of the relation $\RR$ mapping vertices with the same colour and value. However, $\C$ and $\C'$ are not bisimilar, because no facet of $\C'$ is bisimilar to $F_4$ in $\C$.
\begin{center}
\scalebox{.8}{
$\C:$ \quad \quad 
\begin{tikzpicture}
\node (a0) at (0,0) {$a0$};
\node (b0) at (4,0) {$b0$};
\node (c0) at (2,3.42) {$c0$};
\node (a1) at (3,1.71) {$a1$};
\node (b1) at (1,1.71) {$b1$};
\node (c1) at (2,0) {$c1$};
\node (f1) at (1,.65) {$F_1$};
\node (f2) at (2,2.36) {$F_2$};
\node (f3) at (3,.65) {$F_3$};
\node (f4) at (2,1.05) {$F_4$};
\draw[-] (a0) -- (c1);
\draw[-] (a0) -- (b1);
\draw[-] (a1) -- (c0);
\draw[-] (a1) -- (c1);
\draw[-] (a1) -- (b0);
\draw[-] (b0) -- (c1);
\draw[-] (b1) -- (c1);
\draw[-] (b1) -- (c0);
\draw[-] (b1) -- (a1);
\end{tikzpicture}
\qquad \qquad
$\C':$ \quad \quad 
\begin{tikzpicture}
\node (a0) at (0,0) {$a0$};
\node (b0) at (4,0) {$b0$};
\node (c0) at (2,3.42) {$c0$};
\node (a1) at (3,1.71) {$a1$};
\node (b1) at (1,1.71) {$b1$};
\node (c1) at (2,0) {$c1$};
\node (f1) at (1,.65) {$F_1$};
\node (f2) at (2,2.36) {$F_2$};
\node (f3) at (3,.65) {$F_3$};
\node (f4) at (2,1.05) {$\epsilon$};
\draw[-] (a0) -- (c1);
\draw[-] (a0) -- (b1);
\draw[-] (a1) -- (c0);
\draw[-] (a1) -- (c1);
\draw[-] (a1) -- (b0);
\draw[-] (b0) -- (c1);
\draw[-] (b1) -- (c1);
\draw[-] (b1) -- (c0);
\draw[-] (b1) -- (a1);
\end{tikzpicture}
}
\end{center}
\end{example}

\begin{example} \label{example.facevert}
Consider three one-dimensional complexes $\C = (C,\chi,\ell)$, $\C' = (C',\chi',\ell')$, and $\C'' = (C'',\chi'',\ell'')$ for two agents $a,b$, where we assume that $a$ has the same value everywhere, and that $b$ also has the same value everywhere. 
\begin{center}
\scalebox{.8}{
$\C:$ \quad 
\begin{tikzpicture}
\node (a0) at (0,0) {$a$};
\node (b0) at (2,0) {$b$};
\node (a1) at (2,2) {$a$};
\node (b1) at (0,2) {$b$};
\draw[-] (a0) -- node[below] {\scriptsize $Z$} (b0);
\draw[-] (b0) -- node[right] {\scriptsize $Y$} (a1);
\draw[-] (a1) -- node[above] {\scriptsize $X$} (b1);
\draw[-] (b1) -- node[left] {\scriptsize $W$} (a0);
\end{tikzpicture}
\qquad\qquad
$\C':$ \quad 
\begin{tikzpicture}
\node (a0) at (0,0) {$a$};
\node (b0) at (2,0) {$b$};
\node (a1) at (3.5,1) {$a$};
\node (b1) at (2,2) {$b$};
\node (a2) at (0,2) {$a$};
\node (b2) at (-1.5,1) {$b$};
\draw[-] (a0) -- node[below] {\scriptsize $Y'$} (b0);
\draw[-] (b0) -- node[right] {\scriptsize $X'$} (a1);
\draw[-] (a1) -- node[above] {\scriptsize $W'$} (b1);
\draw[-] (b1) -- node[above] {\scriptsize $V'$} (a2);
\draw[-] (a2) -- node[left] {\scriptsize $U'$} (b2);
\draw[-] (b2) -- node[left] {\scriptsize $Z'$} (a0);
\end{tikzpicture}
\qquad\qquad
$\C'':$ \quad  
\begin{tikzpicture}
\node (a0) at (0,0) {$a$};
\node (b0) at (2,0) {$b$};
\draw[-] (b0) -- node[below] {\scriptsize $X''$} (a0);
\end{tikzpicture}
}
\end{center}
Clearly, they are all bisimilar. Two different bisimulations between $\C$ and $\C'$ are:
\[\begin{array}{lcl}
\RR &:=& \{ (X,V'), (X,X'), (X,Z'), (Y,U'), (Y,W'), (Y,Y'), \\ && (Z,V'), (Z,X'), (Z,Z'), (W,U'), (W,W'), (W,Y') \} \\
\RR' &:=& \FF(C) \times \FF(C')
\end{array}\]
Relations $\RR$ and $\RR'$ induce the same relation between vertices, namely relating every $a$ vertex in $\C$ to every $a$ vertex in $\C'$, and every $b$ vertex in $\C$ to every $b$ vertex in $\C'$. This demonstrates that the same relation between vertices may be consistent with different relations between facets. The relation $\RR'$ is the maximal bisimulation between $\C$ and $\C'$. Assuming the relation by vertices as primitive, we can define also define $\RR'$ as: for all facets $X$ and $X'$, ${\RR'}XX'$ iff (${\RR'}X_aX'_a$ and ${\RR'}X_bX'_b$).

Both $\C$ and $\C'$ are bisimilar to the single-edged complex $\C''$. This complex is bisimulation minimal ($\C''$ can be seen as the quotient of $\C$, and of $\C'$, with respect to bisimilarity).
\end{example}

\paragraph*{Subdivision}
Subdivisions typically do not result in bisimilar complexes. Consider the two simplicial models for agents $a,b$ below, where the edge with labels $1$ is subdivided into three edges, while the other two edges of this simplicial model are not subdivided. These two simplicial models are not bisimilar. Clearly, because of {\bf atoms}, the facet $a1$||$b1$ on the left must be in relation with a facet $a1$||$b1$ on the right. Which one? If we choose $a1$||$\pmb{b1}$ on the right, then {\bf forth} fails because the adjoining $b1$||$a2$ on the left cannot be related to $\pmb{b1}$||$\pmb{a1}$ on the right, as the values of the $a$ vertices are different. If we choose $\pmb{b1}$||$\pmb{a1}$ on the right, then {\bf forth} fails because the adjoining $b0$||$a1$ on the left cannot be related to $\pmb{a1}$||$b1$ on the right, as the values of the $b$ vertices are different. Finally and again similarly, if we choose $\pmb{a1}$||$b1$ on the right, then {\bf back} fails, as $b0$||$a1$ and $\pmb{b1}$||$a1$ are not related.

We can also come to this conclusion in a different way: if the respective complexes have different properties of knowledge, then they are not bisimilar (Proposition~\ref{prop.biseq}, later). We now note, for example, that in the complex on the right in vertex $\pmb{b1}$ agent $b$ knows that the value of $a$ is $1$, whereas on the left agent $b$ is uncertain about the value of $a$ if his value is $1$. Similarly, $\pmb{a1}$ on the right is the unique vertex where $a$ knows that the value of $b$ is $1$. 
\begin{center}
\scalebox{.8}{
\begin{tikzpicture}
\node (b0) at (0,0) {$b0$};
\node (a1) at (2,0) {$a1$};
\node (b1) at (4,0) {$b1$};
\node (a2) at (6,0) {$a2$};
\draw[-] (b0) -- (a1);
\draw[-] (a1) -- (b1);
\draw[-] (b1) -- (a2);
\end{tikzpicture}
\qquad {\Large $\not{\!\!\!\bisim}$} \qquad
\begin{tikzpicture}
\node (b0) at (0,0) {$b0$};
\node (a1) at (2,0) {$a1$};
\node (b1) at (4,0) {$\pmb{b1}$};
\node (a2) at (6,0) {$\pmb{a1}$};
\node (b2) at (8,0) {$b1$};
\node (a3) at (10,0) {$a2$};
\draw[-] (b0) -- (a1);
\draw[-] (a1) -- (b1);
\draw[-] (b1) -- (a2);
\draw[-] (a2) -- (b2);
\draw[-] (b2) -- (a3);
\end{tikzpicture}
}
\end{center}
However, sometimes subdivisions result in bisimilar complexes. For example, if, instead, $a1$||$b1$ were the unique edge of a simplicial model, then its subdivision must be bisimilar, as below:
\begin{center}
\scalebox{.8}{
\begin{tikzpicture}
\node (a1) at (2,0) {$a1$};
\node (b1) at (4,0) {$b1$};
\draw[-] (a1) -- (b1);
\end{tikzpicture}
\qquad {\Large $\bisim$} \qquad
\begin{tikzpicture}
\node (a1) at (2,0) {$a1$};
\node (b1) at (4,0) {$\pmb{b1}$};
\node (a2) at (6,0) {$\pmb{a1}$};
\node (b2) at (8,0) {$b1$};
\draw[-] (a1) -- (b1);
\draw[-] (b1) -- (a2);
\draw[-] (a2) -- (b2);
\end{tikzpicture}
}
\end{center}
Similarly, for dimension $2$, the left complex below consists of a single facet, and it is bisimilar to its subdivision on the right (by connecting vertices with the same colour), where we assume that all vertices of the same colour satisfy the same local variable(s). 
\begin{center}
\scalebox{.7}{
\begin{tikzpicture}
\node (a0) at (0,0) {$a$};
\node (b0) at (2,0) {$b$};
\node (c0) at (1,1.71) {$c$};
\draw[-] (a0) -- (c0);
\draw[-] (a0) -- (b0);
\draw[-] (b0) -- (c0);
\node (bisim) at (3.2,.75) {\Large $\bisim$};
\end{tikzpicture}
\quad
\begin{tikzpicture}
\node (a0) at (0,0) {$a$};
\node (b0) at (2,0) {$b$};
\node (a1) at (4,0) {$a$};
\node (b1) at (6,0) {$b$};
\node (c0) at (1,1.71) {$c$};
\node (b3) at (2.5,2.2) {$b$};
\node (a3) at (3.5,2.2) {$a$};
\node (c3) at (3,1.24) {$c$};
\node (c1) at (5,1.71) {$c$};
\node (a2) at (2,3.42) {$a$};
\node (b2) at (4,3.42) {$b$};
\node (c2) at (3,5.13) {$c$};
\draw[-] (a0) -- (b0);
\draw[-] (b0) -- (a1);
\draw[-] (a1) -- (b1);
\draw[-] (c0) -- (b3);
\draw[-] (c1) -- (a3);
\draw[-] (c3) -- (b3);
\draw[-] (a3) -- (b3);
\draw[-] (a3) -- (c3);
\draw[-] (c2) -- (b3);
\draw[-] (c2) -- (a3);
\draw[-] (a0) -- (c0);
\draw[-] (c0) -- (a2);
\draw[-] (a2) -- (c2);
\draw[-] (c2) -- (b2);
\draw[-] (b2) -- (c1);
\draw[-] (c1) -- (b1);
\draw[-] (a0) -- (b3);
\draw[-] (a0) -- (c3);
\draw[-] (a2) -- (b3);
\draw[-] (c3) -- (a1);
\draw[-] (b0) -- (c3);
\draw[-] (a3) -- (b2);
\draw[-] (b1) -- (c3);
\draw[-] (b1) -- (a3);
\end{tikzpicture}
}
\end{center}
Bisimulation may be a useful tool to determine whether seemingly similar simplicial complexes contain the same information. As well known, this can be determined in quadratic time with respect to the size of the models (i.e., complexes). Another use may be to determine whether, given some initial complex, clearly different subdivisions (where agents have different knowledge, such as $a$ knows $p$ in one but never in the other, and $b$ knows $q$ in the other but never in the one) both have subsequent subdivisions resulting, after all, in bisimilar complexes again. This is the property known as confluence, or Church-Rosser. Isomorphy is then often a bridge too far, but bisimilarity may be all we need.

\medskip

We close this section with some elementary results for bisimulations between complexes. We recall the maps $\kappa: \mathcal{S} \imp \mathcal{K}$ and $\sigma: \mathcal{K} \imp \mathcal{S}$ between the simplicial models and the proper local epistemic models. We now have the following.

\begin{proposition}\label{prop:bisimEquivMods}
Let $\model,\model' \in \mathcal K$, let $\C, \C' \in \mathcal S$. Then:
\begin{itemize}
\item If $\model \bisim \model'$, then $\sigma(\model) \bisim \sigma(\model')$.
\item If $\C \bisim \C'$, then $\kappa(\C) \bisim \kappa(\C')$.
\item If $\model \bisim \model'$, then $\kappa(\sigma(\model)) \bisim \kappa(\sigma(\model'))$.
\item If $\C \bisim \C'$, then $\sigma(\kappa(\C)) \bisim \sigma(\kappa(\C'))$.
\end{itemize}
\end{proposition}
\begin{proof}
We show the first item. Let $\model = (S,\sim,L)$ and $\model' = (S',\sim',L')$ be local proper epistemic models.

Let $\RR: \model \bisim \model'$. We define relation $\RR'$ beween the facets $\sigma(s) = \{[s]_a \mid a \in A\}$ of $\sigma(\model)$ and the facets $\sigma(s') = \{[s']_a \mid a \in A\}$ of $\sigma(\model')$ as: $\RR'\sigma(s)\sigma(s')$ iff ${\RR}ss'$. Note that this induces that for all $a \in A$, $\RR' [s]_a[s']_a$ iff $\RR ss'$. 

We now show that $\RR'$ is a bisimulation. Let $\RR'\sigma(s)\sigma(s')$.

The clause {\bf atoms} holds: For all $[s]_a \in \sigma(s)$, $[s']_a \in \sigma(s')$ and $p_a \in P$, we have that $p_a \in \ell([s]_a)$ iff $p_a \in L(s)$ iff (using {\bf atoms} for $\RR$, given that $\RR ss'$ follows from $\RR'\sigma(s)\sigma(s')$) $p_a \in L'(s')$ iff $p_a \in \ell([s]_a)$.

The clause {\bf forth} holds: Let $\sigma(t) \in \FF(\sigma(\model))$ such that $a \in \chi(\sigma(s) \inter \sigma(t))$. From the definition of $\sigma(\model)$ we get that $s \sim_a t$. Therefore, given $\RR'\sigma(s)\sigma(s')$ and therefore ${\RR}ss'$, and as $\RR$ is a bisimulation, there is $t' \in \model'$ such that $s' \sim_a t'$ and $\RR tt'$. We now choose $\sigma(t') \in \FF(\sigma(\model'))$ to obtain $\RR' \sigma(t)\sigma(t')$ and $a \in \chi(\sigma(s') \inter \sigma(t'))$, as required.

The {\bf back} step  is similar.

The proof of the second item is also fairly similar.

The last two items follow from the first two items, but also directly from $\model \simeq \kappa(\sigma(\model))$ and $\C \simeq \sigma(\kappa(\model))$: isomorphy is stronger than bisimilarity, and bisimilarity is transitive.
\end{proof}
These results extend the known results that $\kappa(\sigma(\model)) \simeq \model$ and that $\sigma(\kappa(\C)) \simeq \C$ from~\cite{goubaultetal:2018} (Prop.~\ref{prop.corr}) in the expected way. Clearly, we also have such correspondence between pointed simplicial complexes, for example, $\sigma(\kappa(\C,X)) \bisim \sigma(\kappa(\C',X'))$, where $X$ is a facet in the complex of $\C$ and $X'$ is a facet in the complex of $\C'$.

Also for simplicial complexes, bisimilarity implies modal equivalence, and for finite complexes the implication also holds in the other direction, as is to be expected. 
It is often assumed in combinatorial topology and in distributed computing that complexes are finite, but not always. Interesting applications exist for the infinite case~\cite{AguileraInf04}. 
\begin{proposition} \label{prop.biseq} Let $(\C,X)$ and $(\C',X')$ be given, with $\C = (C,\chi,\ell)$ and $\C' = (C',\chi',\ell')$. Then:
\begin{itemize}
\item $(\C,X) \bisim (\C',X')$ implies $(\C,X) \equiv (\C',X')$. 
\end{itemize}
Let now, additionally, the sets of vertices of $\C$ and $\C'$ be finite. Then:
\begin{itemize}
\item $(\C,X) \equiv (\C',X')$ implies $(\C,X) \bisim (\C',X')$. 
\end{itemize} \vspace{-.5cm}
\end{proposition}
\begin{proof}
These proofs are elementary. Let us show the second item.

Define relation $\RR:\FF(C)\times\FF(C')$ as: for all $X \in \FF(C)$ and $X' \in \FF(C')$, ${\RR}XX'$ if $(\C,X)\equiv(\C',X')$. We show that $\RR$ is a bisimulation. 

Let ${\RR}XX'$ be arbitrary.

First consider condition {\bf atoms}. For all $p_a\in P$, $p_a \in\ell(X)$ iff $p_a \in \ell'(X')$, because $(\C,X)\equiv(\C',X')$ implies that $(\C,X)\models p_a$ iff $\C',X'\models p_a$. 

Now consider {\bf forth}. Let $Y \in \FF(C)$ be such that $a \in \chi(X \inter Y)$. Consider $\pmb{Y'} := \{ Y' \in \FF(C') \mid  a \in \chi(X' \inter Y')\}$. A $\C'$ is finite, $\pmb{Y'}$ is finite, let $\pmb{Y'} = \{Y'_1,\dots,Y'_n\}$. Now suppose towards a contradiction that $Y$ is not in the relation $\RR$ with any of the $Y'_1, \dots, Y'_n$, so that therefore $(\C,Y)\not\equiv(\C',Y'_1)$, \dots, $(\C,Y)\not\equiv(\C',Y'_1)$. Then there are $\phi_1,\dots,\phi_n$ such that $\C',Y'_1 \models \phi_1$ but $\C,Y \not\models \phi_1$, \dots, $\C',Y'_n \models \phi_n$ but $\C,Y \not\models \phi_n$. Using the semantics of knowledge, it follows that $\C',X'\models K_a (\phi_1 \vel \dots \vel \phi_n)$ whereas $\C,X\not\models K_a (\phi_1 \vel \dots \vel \phi_n)$. Therefore it is not the case that $(\C,X)\equiv(\C',X')$. This contradicts the assumption ${\RR}XX'$. 

The condition {\bf back} is shown similarly.
\end{proof}
Given complexes $C$ and $C'$, let $f: \VV(C) \imp \VV(C')$ be a chromatic simplicial map. The relation $\RR \subseteq \VV(C) \times \VV(C')$ defined as, for all $x \in \VV(C)$ and $x' \in \VV(C')$, $\RR\, xy$ iff $f(x)=y$ is called the \emph{induced} relation. Also, as $\RR$ is functional, this, in its turn, induces a unique relation $\RR$ between facets (in the usual way). We recall that a simplicial map is value preserving between simplicial models iff for all $v \in \VV(C)$, $\ell'(f(v)) = \ell(v)$ (where $\ell$ is the valuation on $C$ and $\ell'$ that on $C'$). We now have that:

\begin{proposition}\label{prop:sim}
The induced relation $\RR$ of a value preserving chromatic simplicial map~$f$ between simplicial models $\C$ and $\C'$ is a simulation.
\end{proposition}
\begin{proof}
The proof is elementary.
\end{proof}
As a direct consequence of Proposition~\ref{prop:sim}, if both $f$ and $f^{-1}$ are value preserving chromatic simplicial maps, then the induced relation $\RR$, that determines a unique relation between facets as it is functional, is a bisimulation between $\C$ and $\C'$. This bisimulation is an isomorphism.

\paragraph*{Covering complex}
As explained by Hatcher~\cite{Hatcher}, algebraic topology can be roughly defined as the study of techniques for forming algebraic images of topological spaces.
Thus, continuous maps between spaces are projected onto homomorphisms between their algebraic images, so topologically related spaces have algebraically related images.
One of the simplest and most important functors of algebraic topology, the \emph{fundamental group,} which creates an algebraic image of a space from the loops in the space (i.e., paths starting and ending at the same point).
The fundamental group of a space $\pmb{X}$ is defined so that its elements are loops in $\pmb{X}$ starting and ending at a fixed basepoint $x_0\in \pmb{X}$ , but two such loops are regarded as determining the same element of the fundamental group if one loop can be continuously deformed into the other within the space $\pmb{X}$. 
There is a very deep connection between the fundamental group and \emph{covering spaces}, and  they can be regarded as two viewpoints toward the same thing~\cite{coveringRotman}. This means that algebraic features of the fundamental group can often be translated into the geometric language of covering spaces.
Roughly speaking, a space $\pmb{Y}$ is called a covering space of $\pmb{X}$ if $\pmb{Y}$ maps onto $\pmb{X}$ in a locally
homeomorphic way, so that the pre-image of every point in $\pmb{X}$ has the same cardinality.

We consider the combinatorial version of covering complex~\cite{coveringRotman}, extended to the epistemic setting.
Let  $\C = (C,\chi,\ell)$ and $\C' = (C',\chi',\ell')$ be simplicial models.
An \emph{epistemic covering complex} of $\C'$ is a pair $(\C,f)$ where $C$ is a connected 
complex and $f:\C\imp \C'$ is a value preserving simplical map such that for every simplex $X$ in $C'$, $f^{-1}(X)$ is a union of pairwise 
disjoint simplices. We observe that $C'$ is the image of $f$. Hence, $C'$ is also connected. 

Applications of covering spaces to distributed computing are explored in~\cite{FraigniaudRT13}, and this is an example.
In the figure we assume that a vertex $v$ and $f(v)$ have the same valuations, and an edge in $C$ is sent by $f$ to the edge of the same label in $C'$.

\begin{center}
\includegraphics[scale=0.45]{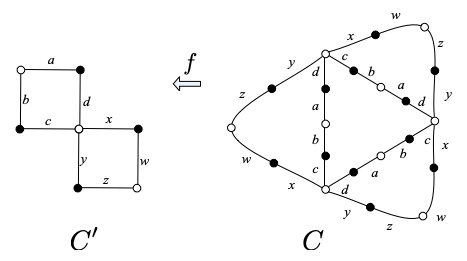}
\label{fig:covering}
\end{center}

\begin{proposition}\label{prop:coveringMain}
If $(\C,f)$ is an epistemic covering complex of $\C'$ then there is a total bisimulation between $\C$ and $\C'$.
\end{proposition}

\begin{proof}
By Proposition~\ref{prop:sim} the induced relation $\RR$ of   $f:\C\imp \C'$ is a  {simulation} from $\C$ to~$\C'$.
It is easy to see that the {\bf back} property holds as well.
\end{proof}

\section{Local semantics for simplicial complexes} \label{section.local}


In \cite{goubaultetal:2018} the semantics of knowledge on simplicial complexes employs a satisfaction relation between a pair consisting of a simplicial model and a facet of that model, and a formula. However, in combinatorial topology we not only wish to determine what is true in a facet, but also what is true in any simplex. We wish a more \emph{local} semantics. There are two roads towards such local semantics for simplicial complexes. There is another road from a local semantics for simplices to the special case of facets.

\subsection{From semantics for facets to semantics for simplices}

\paragraph*{Multi-pointed local semantics} In the first place, just as in Kripke semantics, we can define the satisfaction relation between multi-pointed simplicial models and formulas, instead of between pointed simplicial complexes and formulas, using that for any face of dimension $m < n$, there is a unique set of facets that contain it (we recall that we only consider pure complexes).

A multi-pointed simplicial model $(\C, \pmb{X})$ is a pair consisting of a simplicial model $\C = (C,\chi,\ell)$ and a set $\pmb{X}$ of facets $X \in \FF(C)$. The satisfaction relation for such sets of facets is easily defined in the obvious way.
\[ \begin{array}{lcl}
\C,\pmb{X} \models \phi & \text{iff} & \C,X \models \phi \ \text{for all} \ X \in \pmb{X}
\end{array}\]
Consequently, given any simplex $Y \subseteq C$ of $\C$, and the set $\pmb{X} := \{ X \in \FF(C) \mid Y \subseteq X \}$ of facets containing it as a face, we can define the local truth in $Y$ as follows.
\[ \begin{array}{lcl}
\C,Y \models \phi & \text{iff} & \C,X \models \phi \ \text{for all} \ X \in \FF(C) \ \text{with} \ Y \subseteq X \end{array}\]

\begin{example} \label{example.local}
Consider the simplicial model depicted below, where, as before, for the convenience of the exposition the vertices have been labelled with the colour/agent $i$ and the value of the single local proposition $p_i$ for that agent, where $0$ means that $p_i$ is false and $1$ means that $p_i$ is true, and where as names of these vertices we use such pairs. Let further $X = \{a0,c1,b1\}$, $Y = \{a1,c1,b1\}$, and $Z = \{a1,c1,b0\}$ be the three facets of the complex, as also suggestively depicted. 
\begin{center}
\scalebox{.8}{
\begin{tikzpicture}
\node (a0) at (0,0) {$a0$};
\node (b0) at (4,0) {$b0$};
\node (a1) at (3,1.71) {$a1$};
\node (b1) at (1,1.71) {$b1$};
\node (c1) at (2,0) {$c1$};
\node (f1) at (1,.65) {$X$};
\node (f3) at (3,.65) {$Z$};
\node (f4) at (2,1.05) {$Y$};
\draw[-] (a0) -- (c1);
\draw[-] (a0) -- (b1);
\draw[-] (a1) -- (c1);
\draw[-] (a1) -- (b0);
\draw[-] (b0) -- (c1);
\draw[-] (b1) -- (c1);
\draw[-] (b1) -- (a1);
\end{tikzpicture}
}
\end{center}
We now have that $\C,X \models K_a \neg p_a$, but $\C,\{b1,c1\} \not \models K_a \neg p_a$. As $\{b1,c1\} = X \inter Y$, $\C,\{b1,c1\} \models \phi$ is equivalent to ($\C,X \models \phi$ and $\C,Y \models \phi$). And from $\C,Y \models K_a p_a$ it follows that $\C,Y \not\models K_a \neg p_a$, and thus $\C,\{b1,c1\} \not \models K_a \neg p_a$. The edge $\{b1,c1\}$ represents that agents $b$ and $c$ are uncertain about what $a$ knows.

Similarly, in vertex $c1$, agent $c$, who there knows that $p_c$, is uncertain between valuations $\neg p_a \et p_b \et p_c$, $p_a \et p_b \et p_c$, and $p_a \et \neg p_b \et p_c$. This is described by $\C, c1 \models K_c ((\neg p_a \et p_b \et p_c) \vel (p_a \et p_b \et p_c) \vel (p_a \et \neg p_b \et p_c)) \et \M_c (\neg p_a \et p_b \et p_c) \et \M_c (p_a \et p_b \et p_c) \et \M_c (p_a \et \neg p_b \et p_c)$. In particular, from this it also follows that $\C, c1 \models K_c (p_a \vel p_b)$. Although $c$ considers it possible that $p_a$ is false, and also considers it possible that $p_b$ is false, she does not consider it possible that both are false.
\end{example}

\paragraph*{Language-restricted local semantics} There is also another, slightly more winding, road to a local semantics for simplicial models. This involves restricting the language as well: when interpreting a formula in a simplex, we then only allow knowledge modalities and local variables for the colours of that simplex. Given the language $\lang_K(A,P)$, where $P = \Union_{a \in A} P_a$, and given a simplicial model $\C = (C,\chi,\ell)$ and a simplex $X \in C$, satisfaction $\C,X \models \phi$ is only defined for formulas $\phi$ in the language restricted to agents in $\chi(X)$ and to local variables for those agents. In other words, we then only formalize what the agents in $\chi(X)$ know and do not know about themselves and about their own local variables, and not what they know about other agents or their local variables, or about what those other agents know. (Similar restrictions abound in game theory and computational social choice \cite{brandtetal:2016}.) This comes with the following semantics.

The relation $\models$ between a pair $(\C,X)$ and $\phi$ is defined for $\phi \in \lang_K(\chi(X),P|\chi(X))$, where $P|\chi(X) := \Union_{a \in \chi(X)} P_a$, by induction on the structure of $\phi$. The non-obvious two clauses are as follows. The language restriction implies that, below, the $a$ in $p_a$ and in $K_a$ must be in the set of colours $\chi(X)$.
\[ \begin{array}{lcl}
\C, X \models p_a & \text{iff} & p_a \in \ell(X) \\
\C,X \models K_a\phi & \text{iff} & \C,Y \models \phi \ \text{for all} \ X \in C \ \text{with} \ \chi(X) = \chi(Y) \ \text{and} \ a \in \chi(X \inter Y)
\end{array}\]

Instead of defining this syntactically, we can also define this semantically, as follows.

Let $\C = (C,\chi,\ell)$ be a simplicial model for agents $A$, and $A' \subseteq A$ a subset of agents. The {\em restriction $\C|A'$ of $\C$ to $A'$} is the simplicial model obtained from $\C$ by keeping only the vertices coloured by $A'$, i.e., $\C|A' = (C',\chi',\ell')$ where $C' = \{ X \in C \mid \chi(X) \subseteq A' \}$, where for all $X \in C'$, $\chi'(X) = \chi(X)$, and where for all $p_a \in P_a$ with $a \in A'$ and for all $v \in \VV(C')$, $p_a \in \ell'(v)$ iff $p_a \in \ell(v)$. We note that $C'$ is again a pure chromatic simplicial complex and that its dimension is $|A'|-1$. The language restricted local semantics above is just the regular semantics in $\C|A'$, for formulas in the language $\lang_K(A', P|A')$. In other words, for all $\phi \in \lang_K(A', P|A')$: \[ \C,X \models \phi \quad \text{iff} \quad \C|\chi(X), X \models \phi \]

\begin{example} \label{example.local2}
Reconsider Example~\ref{example.local}. We now have that $\C,\{b1,c1\} \models K_a \neg p_a$ and $\C, c1 \models K_c (p_a \vel p_b)$ are undefined, because $a \notin \chi(\{b1,c1\}) = \{b,c\}$ respectively $a \notin \chi(\{c1\})$, i.e., $a \neq c$. However, we have that $\C, \{b1,c1\} \models K_b p_c \et \neg K_c p_b$ and, less interestingly, as this is `all that $c$ knows', that $\C,c1 \models K_c p_c$.
\end{example}

\weg{
\begin{example} \label{example.bigtri}
Consider simplicial model $C$ below. Given agents $a,b,c$ let their unique local variable be $p_a,p_b,p_c$; below, $1$ after an agent name means that the value of the local variable for that agent is true, and $0$ that it is false.
\begin{center}
\scalebox{.8}{
\begin{tikzpicture}
\node (a0) at (0,0) {$a0$};
\node (b0) at (4,0) {$b0$};
\node (c0) at (2,3.42) {$c0$};
\node (a1) at (3,1.71) {$a1$};
\node (b1) at (1,1.71) {$b1$};
\node (c1) at (2,0) {$c1$};
\node (f1) at (1,.65) {$F_1$};
\node (f2) at (2,2.36) {$F_2$};
\node (f3) at (3,.65) {$F_3$};
\node (f4) at (2,1.05) {$F_4$};
\draw[-] (a0) -- (c1);
\draw[-] (a0) -- (b1);
\draw[-] (a1) -- (c0);
\draw[-] (a1) -- (c1);
\draw[-] (a1) -- (b0);
\draw[-] (b0) -- (c1);
\draw[-] (b1) -- (c1);
\draw[-] (b1) -- (c0);
\draw[-] (b1) -- (a1);
\end{tikzpicture}
}
\end{center}
We can now evaluate various propositions. For simplicity we equate edges and vertices with the colour/value pairs labelling them.
\begin{itemize}
\item $C,F_1 \models K_a p_b$ whereas $C,F_1 \models \neg (K_b \neg p_a \vel K_b p_a)$ ; $C,a0c1 \models K_a p_b$ and $C,a0c1 \models K_b p_b$ are undefined, as the interpretation is now local: restricted to the language for agents $a$ and $c$; on the other hand we have $C,a0c1 \models K_a p_c$.
\item $C,F_1 \models K_a p_b$
\end{itemize}
\end{example}
}

Let us write $\models_{mp}$ for the satisfaction relation in the `multi-pointed local semantics' and $\models_{lr}$ for that in the `language-restricted local semantics'. Then truth with respect to the latter is preserved as truth with respect the former, which is desirable.
\begin{proposition}
Let $\phi \in \lang_K(X)$. Then $\C, X \models_{lr} \phi$ iff $\C, X \models_{mp} \phi$.
\end{proposition}
\begin{proof}
The proof is by induction on $\phi$ where the only relevant case is the one for knowledge:

Let us assume that $\C, X \models_{lr} K_a \phi$, i.e., for all $Y \in C$ with $\chi(X) = \chi(Y)$ and $a \in \chi(X \inter Y)$, $\C, Y \models_{lr} \phi$. We wish to prove that $\C, X \models_{mp} K_a \phi$. In order to prove that, using the $\models_{mp}$ semantics, assume that $Z \in \FF(C)$ with $X \subseteq Z$, and, using the semantics for knowledge, let $V \in \FF(C)$ with $a \in \chi(Z \inter V)$. It then remains to show that $\C, V \models_{mp} \phi$. Let now $Y \subseteq V$ with $\chi(Y) = \chi(X)$. Note that, as $a \in \chi(X)$, therefore also $a \in \chi(Y)$. From $X \subseteq Z$, $Y \subseteq V$, $a \in \chi(Z \inter V)$, $a \in \chi(X)$, and $a \in \chi(Y)$ it follows that $a \in \chi(X \inter Y)$. From that, the initial assumption $\C, X \models_{lr} K_a \phi$ and the $\models_{lr}$ semantics it now follows that $\C, Y \models_{lr} \phi$. By inductive hypothesis it now follows that $\C, Y \models_{mp} \phi$, and from the $\models_{mp}$ semantics and $Y \subseteq V \in \FF(C)$ it then follows that $\C, V \models_{mp} \phi$, as required.

Conversely, assume that $\C, X \models_{mp} K_a\phi$. Suppose that $Y \in C$ with $\chi(X) = \chi(Y)$ and $a \in \chi(X \inter Y)$. Let $Z \in \FF(C)$ with $X \subseteq Z$, then from the assumption $\C, X \models_{mp} K_a\phi$ we obtain $\C, Z \models_{mp} K_a\phi$. Therefore, for all $V \in \FF(C)$ with $a \in \chi(Z\inter V)$, $\C, V \models_{mp} \phi$. Similar to the reasoning in the other direction, for any $Y \subseteq V$ with $\chi(Y) = \chi(X)$ and $a \in \chi(Y)$ we must have that $a \in \chi(X \inter Y)$, such that by definition of $\models_{mp}$ we obtain $\C,Y \models_{mp} \phi$, and with induction we can then conclude $\C,Y \models_{lp} \phi$, as required.
\end{proof}

\paragraph*{Local semantics for distributed epistemic models}
One can just as well consider a similar local semantics for distributed epistemic models. 
Let $\model = (S,\sim,L)$ be a distributed model for agents $A$ and variables $P = \Union_{a \in A} P_a$, so that states $s \in S$ have shape $s = (s_{a_0},\dots,s_{a_n})$. Given such $s \in S$, by $s \subseteq_B s'$ we mean that $s$ is the restriction of $s$ to the agents $B \subseteq A$. For example, $(s_{a_1},s_{a_2}) \subseteq_{\{a_1,a_2\}} (s_{a_0},s_{a_1},s_{a_2})$. We now can, similarly to above, either define, in multi-pointed fashion, that $\model, s|B \models \phi$ iff (for all $s' \in S$ with $s'|B = s|B$: $\model,s' \models \phi$), or define, in restricted-language fashion, $\model, s|B \models \phi$ by induction on $\phi \in \lang(B,P|B)$. Again, truth in the latter is preserved as truth in the former.

\subsection{A semantics for simplices including facets}

A more rigorous departure from the semantics where the evaluation point is a facet are the following alternative semantics for arbitrary simplices as evaluation points, of which the case for facets is merely an instantiation.

Given are agents $A$ and variables $P = \Union_{a \in A} P_a$. 
Let $\C = (C,\chi,\ell)$ be a simplicial model, $X \in C$, and $\phi \in \lang(A,P)$. Similar to before for facets, we extend the usage of $\ell$ to include arbitrary simplices as: $p \in \ell(X)$ iff there is a $v \in X$ with $p \in \ell(v)$.

\[ \begin{array}{lcl}
\C, X \models p_a & \text{iff} & p_a \in \ell(X) \\
\C, X \models \phi\et\psi & \text{iff} & \C, X \models \phi \ \text{and} \ \C, X \models \psi \\
\C, X \models \neg \phi & \text{iff} & \C, X \not\models \phi \\
\C,X \models K_a\phi & \text{iff} & a \in \chi(X) \ \text{and} \ \C,Y \models \phi \ \text{for all} \ Y \in \FF(C) \ \text{with} \ a \in \chi(X \inter Y) \\
\end{array}\]
The special case of the knowledge semantics for vertices therefore is:
\[ \begin{array}{lcl}
\C,v \models K_a\phi & \text{iff} & a = \chi(v) \ \text{and} \ \C,Y \models \phi \ \text{for all} \ Y \in \FF(C) \ \text{with}\ v \in Y \\
\end{array}\]
\weg{
An alternative formulation of the knowledge semantics is:
\[ \begin{array}{lcl}
\C,X \models K_a\phi & \text{iff} & \text{there is} \ v \in X \ \text{with} \ \chi(v) = a \ \text{and} \ \C,Y \models \phi \ \text{for all} \ Y \in \FF(C) \ \text{with} \ v \in Y \end{array}\]
In other words, $a$ knows $\phi$ in $X$ if $X$ contains a vertex $v$ coloured $a$ and $\chi$ is true in the \emph{star} of $v$. So, in particular, $a$ knows $\phi$ in a vertex $v$ coloured $a$ iff $\phi$ is true in the \emph{star} of $v$. 
}
Elementary results corroborating this more general perspective of this semantics are:

\begin{proposition}\label{prop.star}
If $\C,X \models \phi$ and $Y \in C$ such that $X \subseteq Y$, then $\C,Y \models \phi$.
\end{proposition}
\begin{proposition}
If $\C,X \models \phi$, $Y \in C$ such that $Y \subseteq X$ and $\phi \in \lang(\chi(Y))$, then $\C,Y \models \phi$.
\end{proposition}
Both are easily shown by induction on $\phi$. Proposition~\ref{prop.star} says, in other words, that if $\phi$ is true in simplex $X$ then it is true in the {\em star} of $X$.

\section{From global to local epistemic models} \label{section.globallocal}

Propositions \ref{prop.corr} and \ref{prop.corr2} in Section \ref{section.tools} reported how epistemic models correspond to simplicial models, however on the assumption that the epistemic models are local. We now investigate when epistemic models that are not local can be transformed into local epistemic models that contain the same information. It is easy to see that changing the domain or the epistemic relations of an epistemic model can result in change of information content, so that the only option is to adjust the valuation by replacing the global variables with local variables. 

Given epistemic model $\model = (S,\sim,L)$ for agents $A$ and variables $P$ we therefore wish to construct a local epistemic model $\model' = (S,\sim,L')$ for agents $A$ and variables $P'$ containing the {\em same information}, i.e., for all $\phi \in \lang_K(A,P)$ there is a $\phi'\in \lang_K(A,P')$ such that $\I{\phi}_\model = \I{\phi'}_{\model'}$, and for all $\phi' \in \lang_K(A,P')$ there is a $\phi\in \lang_K(A,P)$ such that $\I{\phi}_\model = \I{\phi'}_{\model'}$.

In general, as known from \cite{goubaultetal:2018}, this quest cannot succeed. A simple counterexample is the {\em improper} epistemic model \[ \np \stackrel a {\text{------}} p \] for one agent $a$ and one variable $p$. Agent $a$ cannot distinguish  two states with different valuations. Any local variable $q_a$ of agent $a$ must be true in both  states, as they are indistinguishable for $a$:
\[ q_a \stackrel a {\text{------}} q_a \] This model collapses to the singleton \[ q_a \]
We have lost information. 

However, on other epistemic models this quest succeeds. For example, consider
the following (proper) model for two agents $a,b$ and one variable $p$.
\[ p \stackrel a {\text{------}} \np \stackrel b {\text{------}} p \]
We note that $p$ is not local for $a$ and also not local for $b$. Now consider:
\[ q_aq_b \stackrel a {\text{------}} q_a\nq_b \stackrel b {\text{------}} \nq_a\nq_b \] for variables $q_a,q_b$. This model is local. As both models are finite and bisimulation minimal, all subsets of the domain are obviously definable in both languages and can thus be made to correspond. We can do this systematically for all formulas by, for example, replacing all occurrences of $p$ by $(q_a\et q_b) \vel (\nq_a\et\nq_b)$ in any formula in the language with $p$, and in the other direction by simultaneously replacing all occurrences of $q_a$ by $\neg K_a p$ and of $q_b$ by $K_b p$ in the language with $q_a,q_b$.

In \cite{ledent:2019} the following method is proposed to make arbitrary finite epistemic models local. Given $\model = (S,\sim,L)$ for variables $P$, consider novel variables $P' = \{p^{[s]_a}_a \mid s \in S, a \in A\}$, i.e., one for each equivalence class $[s]_a$. Then, define $L'(s) = \{ p^{[s]_a}_a \in P' \mid a \in A \}$. The model  $(S,\sim,L')$ is local, and if $\model$ is proper, then $\Inter_{a \in A} [s]_a = s$, so that $\Et_{a \in A} p^{[s]_a}_a$ is a distinguishing formula for state $s$. Therefore, the denotation of $p$ in $\model$ is the same as that of  $\Vel_{p \in L(s)} \Et_{a \in A} p^{[s]_a}_a$ in $\model'$ (see \cite{ledent:2019}). Dually, on bisimulation minimal epistemic models, the denotation of any $p^{[s]_a}_a$ in $\model'$ is the same as the of $\Vel_{t \in [s]_a} \delta_t$ in $\model$, in the original language $\lang_K(P)$. This solution is reminiscent of hybrid logic \cite{Blackburn00} wherein all states have names, and in particular of the enrichment of models in symbolic model checking \cite{EGS18,charrieretal:2015}. This method may result in many more propositional variables. Also, on epistemic models that are not bisimulation minimal, the model with variables $p^{[s]_a}_a$ is not bisimilar to the initial model and then contains more information, as in the following example.

\begin{example}
Consider the three models below. The leftmost model is improper and the other two are proper. The middle model (with names of states in between parentheses) is proper but not bisimulation minimal. The rightmost model is obtained by having variables correspond to equivalence classes. It is bisimulation minimal. For example, the formula $p^{uv}_a \et p^{su}_b$ is only true in state $u$, but any formula in the middle model that is true in $u$ is also true in the bisimilar state $t$.

\begin{center}
\begin{tikzpicture}
\node (00) at (0,0) {\color{white}$p(s)$};
\node (10) at (1.5,0) {\color{white}$\np(t)$};
\node (01) at (0,1.5) {$\np$};
\node (11) at (1.5,1.5) {$p$};
\draw[-] (01) -- node[above] {\scriptsize $ab$} (11);
\end{tikzpicture}
\quad\quad
\begin{tikzpicture}
\node (00) at (0,0) {$p(s)$};
\node (10) at (1.5,0) {$\np(t)$};
\node (01) at (0,1.5) {$\np(u)$};
\node (11) at (1.5,1.5) {$p(v)$};
\draw[-] (00) -- node[above] {\scriptsize $a$} (10);
\draw[-] (01) -- node[above] {\scriptsize $a$} (11);
\draw[-] (00) -- node[left] {\scriptsize $b$} (01);
\draw[-] (10) -- node[right] {\scriptsize $b$} (11);
\end{tikzpicture}
\quad\quad
\begin{tikzpicture}
\node (00) at (0,0) {$p^{st}_ap^{su}_b$};
\node (10) at (2,0) {$p^{st}_ap^{tv}_b$};
\node (01) at (0,1.5) {$p^{uv}_ap^{su}_b$};
\node (11) at (2,1.5) {$p^{uv}_ap^{tv}_b$};
\draw[-] (00) -- node[above] {\scriptsize $a$} (10);
\draw[-] (01) -- node[above] {\scriptsize $a$} (11);
\draw[-] (00) -- node[left] {\scriptsize $b$} (01);
\draw[-] (10) -- node[right] {\scriptsize $b$} (11);
\end{tikzpicture}
\end{center}
\end{example}
As said, we aim at a transformation that is information preserving in both directions. Therefore, the method from \cite{ledent:2019} is not what we want. We therefore investigated alternative ways to make models local. Before we can report on that, we need to introduce an additional technical tool: the distinguishing formula.

\paragraph*{Distinguishing formulas}

Given a subset $S' \subseteq S$ of the domain, a \emph{distinguishing formula} for $S'$ is a formula in some logical language typically denoted $\delta_{S'}$, such that:  for all $s \in S'$, $\model,s \models \phi$ whereas for all $s \not\in S'$, $\model,s \not\models \phi$. A distinguishing formula for a singleton $S' = \{s\}$ is a distinguishing formula for $s$. If $\model$ is finite (a model is finite if its domain is finite) and bisimulation minimal, a distinguishing formula always exists in $\lang_K(A,P)$ for any state. We can construct it as follows  \cite{hvdetal.jlc:2014}.

Let a finite epistemic model $\model = (S,\sim,L)$ be given. Let us a call variable $p \in P$ \emph{redundant} on $\model$ if, whenever $\model$ is a not singleton, $p \in L(s)$ for all $s \in S$ or $p \notin L(s)$ for all $s\in S$ (so the extension of $p$ is $S$ or $\emptyset$); and also if there is a $q \in P$ with $p \neq q$ and with $p \in L(s)$ iff $q \in L(s)$ for all $s \in S$. On a finite model, there is only a finite set $P(\model) \subseteq P$ of non-redundant variables (even if $P$ is countably infinite). Given $s \in S$, let $\tau_s$ be the \emph{factual description} of $s$, i.e., $\tau_s := \Et \{ p \in P(\model) \mid p \in L(s) \} \et \Et \{ \neg p \in P(\model) \mid p \notin L(s) \}$, and consider the following procedure:
\[ \begin{array}{lcl}
\delta_s^0 & := & \tau_s \\
\delta_s^{n+1} &:=& \tau_s \et \Et_{a \in A} \Et_{t \in [s]_a} \M_a \delta_t^n \et \Et_{a \in A} K_a \Vel_{t \in [s]_a} \delta_t^n
\end{array}  \hspace{3cm} \hfill (1) \]
Consider the relation $\RR$ defined by $\RR st$ if for all $n \in \Naturals$, $\model, t \models \delta_s^n$. If the model is finite, the relation $\RR$ is a bisimulation, and for all $n \geq |S|$, $\delta_s^{n+1}$ and $ \delta_s^n$ have the same denotation. We therefore take $\delta_s^{|S|}$ as distinguishing formula $\delta_s$ for state $s$.

References for distinguishing formulas in epistemic logic are \cite{jfak.odds:1998,hvdetal.jlc:2014}.

\paragraph*{Distinguishing formulas with local variables}

Let us now adapt procedure (1) for distinguishing formulas, using a novel set of local propositional variables. Firstly, write $\tau_{[s]_a}$ for  $\Vel_{t \in [s]_a} \tau_t$. For each $\tau_{[s]_a}$ create a novel variable $p_a^{\tau_{[s]_a}}$ with the same denotation on the given model. Note that this variable is local for $a$: $K_a p_a^{\tau_{[s]_a}}$ is true iff $p_a^{\tau_{[s]_a}}$ is true. Consider the following procedure:
\[ \begin{array}{lcl}
\delta_s^0 & := & \Et_{a \in A} p_a^{\tau_{[s]_a}} \\
\delta_s^{n+1} &:=& \Et_{a \in A} p_a^{\tau_{[s]_a}} \et \Et_{a \in A} \Et_{t \in [s]_a} \M_a \delta_t^n \et \Et_{a \in A} K_a \Vel_{t \in [s]_a} \delta_t^n 
\end{array} \hspace{2cm} \hfill (2) \]
The procedure is now in the language with local variables $P' = \{p_a^{\tau_{[s]_a}} \mid s \in S, a \in A\}$. In the sequel, in order to simplify notation, we may write $\phi_a$ for any propositional formula $\phi$ equivalent to $\tau_{[s]_a}$.
\begin{example} \label{example.delta2}
Continuing the above example for three states, the following local model is created. As $p \vel \neg p$ is equivalent to $\T$, we write $\T_a$ instead of $p_a^{p \vel \neg p}$, etc., and we only depict the valuations.
\[ \T_a p_b \stackrel a {\text{------}} \T_a\T_b \stackrel b {\text{------}} p_a\T_b \]
The resulting distinguishing formulas are
\[\begin{array}{lllllllll}
\delta^1_s & = & \T_a \et p_b \\
\delta^2_t & = & \T_a \et \T_b \\
\delta^1_u & = & p_a \et \T_b 
\end{array}\]
Also taking into account that agents know their local variables, an alternative description is:
\[\begin{array}{lllllllll}
\delta^1_s & = & p_b \\
\delta^2_t & = & \M_a p_b \et \M_b p_a \\
\delta^1_u & = & p_a 
\end{array}\]
\end{example}
Recalling the relation defined as: $\RR st$ if for all $n \in \Naturals$, $\model, t \models \delta_s^n$, it is unclear if on any given model it is a bisimulation with the new local variables if and only if it is a bisimulation with the old global variables. We conjecture that on finite models $\RR$ is again a bisimulation, thus achieving our aim to construct a local epistemic model with the same information content as any given epistemic model. We further conjecture that on any given epistemic model where $\RR$ is a bisimulation but not so with the local variables, no local epistemic model exists with the same information content. Example \ref{example.inf} demonstrates this for an infinite model.

\begin{example} \label{example.inf}
Consider the following (image-finite) infinite chain $\model$. 

\[ \np \stackrel a {\text{------}} p \stackrel b {\text{------}} \np \stackrel a {\text{------}} p \stackrel b {\text{------}} \np \stackrel a {\text{------}} p \stackrel b {\text{------}} \dots \ . \]
Let us name the states $0, 1, \dots$ from left to right. The model $\model'$ with local variables is:
\[ \T_a\np_b \stackrel a {\text{------}} \T_a\T_b \stackrel b {\text{------}} \T_a\T_b \stackrel a {\text{------}} \T_a\T_b \stackrel b {\text{------}} \T_a\T_b \stackrel a {\text{------}} \T_a\T_b \stackrel b {\text{------}} \dots \ . \]
State $0$ is the unique state where $K_b \neg p$. All finite subsets of the chain can be distinguished by their relation to this endpoint, and therefore in particular all equivalence classes of relations $\sim_a$ and $\sim_b$, i.e., all local states of $a$ and $b$. As an example, we show how three local $b$-states and three local $a$-states are described in the language for the local variables corresponding to booleans of the single variable $p$. For a local state consisting of states $j$ and $k$ we write $jk$, etc. 
\[ \begin{array}{l|l }
\text{eq.\ class} & \text{distinguishing formula} \\
\hline
{01} & \M_a \np_b \\
{23} & \M_a \M_b \M_a \np_b \et \neg \M_a \np_b \\
{45} & \M_a \M_b \M_a \M_b \M_a \np_b \et \neg \M_a \M_b \M_a \np_b \\
\hline
{0} & \np_b \\
{12} & \M_b \M_a \np_b \et \neg \np_b \\
{34} & \M_b \M_a \M_b \M_a \np_b \et \neg \M_b \M_a \np_b
\end{array} \]
The variable $p$ is true in all odd states. However, this subset of the domain is not definable in the language with local variables. In other words, there is a $\phi\in\lang_K(\{a,b\},\{p\}$, namely $\phi = p$, such that for all $\phi' \in \lang_K(\{a,b\},\{\np_b,\top_a,\top_b\})$, $\I{\phi}_\model \neq \I{\phi'}_{\model'}$. The models $\model$ and $\model'$ do not contain the same information.
\end{example}

\section{Group epistemic notions  for simplicial complexes} \label{section.group}

\subsection{Mutual, common, and distributed knowledge} Apart from individual knowledge, group knowledge plays an important role in epistemic logic. There are three well-known kinds of group knowledge. In the first place, there is the notion representing that \emph{everybody knows} a proposition, also known as \emph{mutual knowledge} \cite{meyeretal:1995,faginetal:1995,osborneetal:1994}.\footnote{There is no agreement on terminology here. The notion that everybody knows $\phi$ (which is unambiguous) is in different communities called: {\em shared} knowledge, {\em mutual} knowledge, {\em general} knowledge; where in some communities some of these terms mean common knowledge instead, creating further confusion.} Then, there is the notion of \emph{common knowledge} of a proposition, which entails not merely that everybody knows it, but also that everybody knows that everybody knows it, and so on \cite{friedell:1969,lewis:1969,hvdetal.ckcb:2009}. Finally, the agents have \emph{distributed knowledge} of a proposition if, intuitively (but not technically), they can get to know it by communication \cite{Hayek_AER45,halpernmoses:1990}.\footnote{The intuition puts one on the wrong foot for distributedly known ignorance: it may be distributed knowledge between $a$ and $b$ that `$p$ is true and  $b$ is ignorant of $p$', but this cannot be made common knowledge between them. If $a$ were to inform $b$ of this, they would then have common knowledge between them of $p$, but they would not have common knowledge between them of `$p$ is true en $b$ does not know $p$'.} This notion is less well studied than the other two. This is somewhat surprising, particularly in the context of distributed systems with communicating agents.

\paragraph*{Syntax} The notions of mutual, common, and distributed knowledge are formalized by operators $E_B$, $C_B$, and $D_B$, respectively. Mutual knowledge $E_B \phi$, for $B \subseteq A$, is typically defined by abbreviation as $\Et_{a \in B} K_a \phi$ and only introduced as a primitive language construct for investigations of succinctness and complexity. Whereas $C_B \phi$ and $D_B \phi$ need to be additional inductive constructs in the logical language in order to be able to give a semantics for these notions. Expression $D_B \phi$ stands for `the group of agents $B$ have distributed knowledge of $\phi$,' and expression $C_B \phi$ stands for `the group of agents $B$ have common knowledge of $\phi$.' Both common knowledge and distributed knowledge are known to enlarge the expressivity of corresponding logics, although in different ways, and to handle distributed knowledge we even need a stronger notion of bisimulation, that may serve us surprisingly well for a generalization of simplicial complexes.

\paragraph*{Semantics}
Given some epistemic model with relations $\sim_a$ for $a \in A$, we define ${\sim^\inter_B} := {\Inter_{a \in B} \sim_a}$, and ${\sim^*_B} := {(\Union_{a \in B} \sim_a)^*}$, where $B \subseteq A$. Both these relations are also equivalence relations. The semantics of common knowledge and of distributed knowledge on epistemic models is now as follows.
\[ \begin{array}{lcl}
\model,s \models D_B \phi & \text{iff} & \model,t \models \phi \ \text{for all}\ t \text{ with } s \sim^\inter_Bt \\
\model,s \models C_B \phi & \text{iff} & \model,t \models \phi \ \text{for all}\ t \text{ with } s \sim^*_Bt
\end{array} \]
It is straightforward to interpret distributed and common knowledge on simplicial models. Let $\C = (C,\chi,\ell)$ be a simplicial model (for $A$ and $P$), and let $X \in \FF(C)$. To interpret common knowledge $C_B$, define operation $*_B$ for all $B \subseteq A$ inductively as: for all $X,Y,Z \in \FF(C)$, $X \in *_B(X)$, and if $Y \in *_B(X)$ and $B \inter \chi(Z \inter Y) \neq \emptyset$, then $Z \in *_B(X)$. Then:
\[ \begin{array}{lcl}
\C,X \models D_B \phi & \text{iff} & \C,Y \models \phi \text{ for all } Y \in \FF(C) \  \text{such that} \ B \subseteq \chi(X \inter Y) \\
\C,X \models C_B \phi & \text{iff} & \C,Y \models \phi \text{ for all } Y \in *_B(X)
\end{array} \]

\begin{example}
We recall the simplicial models $\C$ and $\C'$ from Example \ref{example.ins}. We now evaluate some propositions involving common and distributed knowledge.
\begin{itemize}
\item Let us investigate what $a$ and $b$ commonly know in facet $F_1$ in $\C'$. To determine common knowledge of $a$ and $b$ in $F_1$, the commonly known proposition needs to be true in all $F_i$ bordering on edges of the $ab$-path: $a0$|$b1$|$a1$|$b0$, i.e., in all of $F_1,F_2,F_3$. In other words, it has to be (in this particular example) a model validity.

We now have that $\C',F_1 \models C_{ab} (K_c p_c \imp (K_a \neg p_a \vel K_b \neg p_b))$,  as the formula bound by $C_{ab}$ is true for all $F_1, F_2, F_3$. However, $\C,F_1 \not\models C_{ab} (K_c p_c \imp (K_a \neg p_a \vel K_b \neg p_b))$, as $\C,F_4 \not\models K_c p_c \imp (K_a \neg p_a \vel K_b \neg p_b)$. For another example,  
$\C,F_1 \not\models C_{ab} (K_a \neg p_a \vel K_b \neg p_b)$, namely not in $F_2$ where $c$ knows $\neg p_c$. And for yet another example we have that $\C,F_2 \models K_c (K_a p_a \et K_b p_b)$ whereas $\C,F_2 \not\models C_{ab} (K_a p_a \et K_b p_b)$: in $F_2$ agent $c$ knows more about $K_a p_a \et K_b p_b$ than what agents $a$ and $b$ commonly know there.\footnote{Local semantics can also be given for distributed knowledge and for common knowledge. For example, to determine what $a$ and $b$ commonly know on the edge $\{a0,b1\}$ we need to consider facets bordering on the chain $a0$|$b1$|$a1$|$b0$ only. For example, $\C, \{a0,b1\} \models K_a \neg p_a$ whereas $\C,\{a0,b1\} \not \models C_{ab} \neg p_a$.}
\item In a facet of a simplicial model the agents all together always have distributed knowledge of the valuation in that facet (or of anything else that is true on the facet). For example, $\C,F_4 \models D_{abc} (p_a \et p_b \et p_c)$. On the other hand, $\C,F_4 \not\models D_{ac} (p_a \et p_b \et p_c)$, as $a$ and $c$ are both uncertain about the value of $b$. Still their distributed knowledge in $F_4$ (of positive formulas) is larger than that of each of them separately: the distributed knowledge of $a$ and $c$ in $F_4$ is what is true in $F_4,F_3$; whereas, what $a$ knows in $F_4$ is what is true in $F_2,F_4, F_3$, and what $c$ knows in $F_4$ is what is true in $F_1, F_3, F_4$.
\end{itemize}
\end{example}

\subsection{Bisimulation for distributed knowledge} 

As such, distributed knowledge seems uninteresting on simplicial models, as in any facet of a simplicial model the agents always have distributed knowledge of what is true in that facet (although for proper subsets of all agents distributed knowledge may already be slightly more interesting). Dually, in proper epistemic models, the intersection of all relations is the identity relation (although, also dually, the relation $\sim^\inter_B$ for $B \neq A$ need not be the identity). Somewhat surprisingly, explorations concerning distributed knowledge still seem of large interest to describe the information content of simplicial complexes.
 
%
The natural generalization of bisimulation including relations interpreting distributed knowledge not only contains {\bf forth} and {\bf back} clauses for individual agents but now for any group of agents \cite{Roelofsen07,AgotnesW17}. Given $\RR s s'$, we recall the {\bf forth} clause:
\begin{itemize}
\item {\bf forth}: for all $a \in A$, for all $t$ with $s \sim_a t$, there is a $t'$ with $s'\sim'_at'$ such that $\RR tt'$.
\end{itemize}
We now get:
\begin{itemize}
\item {\bf forth}: for all $B \subseteq A$, for all $t$ with $s \sim^\inter_B t$, there is a $t'$ with $s'\sim'^\inter_B t'$ such that $\RR tt'$.
\end{itemize}
This notion of bisimulation is \emph{stronger} (for a given model it is a more refined relation) than bisimulation for relations $\sim_a$ for individual agents only. See Example \ref{ex.stronger}. This stronger notion may be useful to elicit information from simplical sets, a generalization of simplicial complexes, as illustrated in Example \ref{ex.sets}.
\begin{example} \label{ex.stronger}
The typical example where this stronger bisimulation makes a difference is that it distinguishes the following epistemic models.
\begin{center}
\begin{tikzpicture}
\node (m1) at (-2,0) {$p$};
\node (m2) at (-3.5,0) {$\np$};
\node (00) at (0,0) {$p$};
\node (10) at (1.5,0) {$\np$};
\node (01) at (0,1.5) {$\np$};
\node (11) at (1.5,1.5) {$p$};
\draw[-] (00) -- node[above] {\scriptsize $a$} (10);
\draw[-] (01) -- node[above] {\scriptsize $a$} (11);
\draw[-] (00) -- node[left] {\scriptsize $b$} (01);
\draw[-] (10) -- node[right] {\scriptsize $b$} (11);
\draw[-] (m2) -- node[above] {\scriptsize $ab$} (m1);
\node(modelx) at (-2.75,-1) {$(x)$};
\node(modely) at (.75,-1) {$(y)$};
\end{tikzpicture}
%
\end{center}
In the top-right and bottom-left corner of the square model, agents $a$ and $b$ have distributed knowledge of $p$; we note that ${\sim_a \inter \sim_b}$ is the identity relation. The two models are (standardly) bisimilar, but not with the additional clause for the group of agents $\{a,b\}$: in the left one we can then go from a $p$ to a $\neg p$ state by a $\sim^\inter_{ab}$ link, but not in the right one, so the {\bf forth} clause fails for $\{a,b\}$. 
\end{example}

\begin{example}[Distributed knowledge for simplicial sets] \label{ex.sets}
As the models in the previous example are not local, they do not correspond to a simplicial complex. And as the disjunction of the values in indistinguishable states always is the uninformative value `true', there is no transformation to a local epistemic model (see Section \ref{section.local}). The only obvious representation is simplicial model $(i)$ below, where the name $a01$ suggestively stands for a vertex coloured with $a$ and valued with the indistinguishable set $\{\neg p, p\}$. A similar quadrangular simplicial model represents that four-state epistemic model. But it is bisimilar to~$(i)$.

Distributed knowledge may now be helpful. In model $(y)$, the coalition $\{a,b\}$ (for which we will write $ab$) can distinguish all four states, we have that $p \imp D_{ab}\, p$ and $\neg p \imp D_{ab} \neg p$ are both true. Now consider adding this coalition as a third agent to model $(y)$. Let us call this model $(y3)$. Then the model is proper, it is bisimilar to a two-state model that is like $(x)$ except that the third agent $ab$ can distinguish the $\neg p$ from the $p$ state, let us call that model $(x3)$, and its corresponding simplicial model is $(ii)$ below. Here, agent $ab$ can `inform' agents $a$ and $b$ of the value of $p$. 

Clearly, any simplicial model for $n$ colours/agents can similarly be enriched with additional colours for all subsets of agents to thus obtain a simplicial model for $2^n-1$ agents wherein all distributed knowledge has become explicit. Note that {\em all} epistemic models thus become proper. It is tempting to model this in combinatorial topology with a generalization of simplicial complexes to so-called {\em simplicial sets}, as suggested in \cite[Section 3.7]{ledent:2019}. Just as a complex is a set of subsets of a set of vertices, a simplicial set is a multi-set of its subsets, with the same restriction that it is closed under subsets (see e.g.\ \cite{hilton_wylie_1960}, where a simplicial set is called a \emph{pseudocomplex}). In $(iii)$ and $(iv)$ we tentatively depict $(i)$ and $(ii)$ as simplicial sets.

\begin{center}
\scalebox{.8}{
\begin{tikzpicture}
\node (a02) at (-8,0) {$a01$};
\node (b02) at (-6,0) {$b01$};
\draw[-] (a02) -- (b02);
\node (a01) at (0,0) {$a01$};
\node (b01) at (2,0) {$b01$};
\draw[-,bend right=15] (a01) to (b01);
\draw[-,bend left=15] (a01) to (b01);
\node (a0) at (-4,0) {$a01$};
\node (b0) at (-2,0) {$b01$};
\node (ab0) at (-3,-1.71) {$ab0$};
\node (ab1) at (-3,1.71) {$ab1$};
\draw[-] (a0) -- (b0);
\draw[-] (a0) -- (ab1);
\draw[-] (a0) -- (ab0);
\draw[-] (b0) -- (ab1);
\draw[-] (b0) -- (ab0);
\node (a03) at (4,0) {$a01$};
\node (b03) at (6,0) {$b01$};
\node (ab03) at (5,-1.71) {$ab0$};
\node (ab13) at (5,1.71) {$ab1$};
\draw[-,bend right=15] (a03) to (b03);
\draw[-,bend left=15] (a03) to (b03);
\draw[-] (a03) -- (ab13);
\draw[-] (a03) -- (ab03);
\draw[-] (b03) -- (ab13);
\draw[-] (b03) -- (ab03);
\node (1) at (-7,-2.5) {$(i)$};
\node (2) at (-3,-2.5) {$(ii)$};
\node (3) at (1,-2.5) {$(iii)$};
\node (4) at (5,-2.5) {$(iv)$};
\end{tikzpicture}
}
\end{center}

\end{example}

\subsection{Novel notions of group knowledge for simplicial complexes} 

Beyond applying the standard notions of group knowledge for epistemic models on simplicial models, we can also consider novel notions of group knowledge for simplicial models, for example, to express dimensional properties of complexes. Consider the following notion in between common knowledge and distributed knowledge.

Let $\mathcal B \subseteq \powerset(A)$. We extend the language with an inductive construct $CB_{\mathcal B} \phi$, for ``it is common distributed knowledge that $\phi$ (with respect to ${\mathcal B}$).'' Given epistemic model $\model = (S,\sim,L)$, let ${\sim^{*\inter}_{\mathcal B}} := (\Union_{B \in \mathcal B} \sim^\inter_B)^*$. This relation interprets operator $CB_{\mathcal B} \phi$ as follows. 
\[ \begin{array}{lcl}
\model,s \models \CD_{\mathcal B} \phi & \text{iff} & \model,t \models \phi \ \text{for all}\ t \in S \text{ with } s \sim^{*\inter}_{\mathcal B} t \\
\end{array} \]
Such a semantics has been proposed in \cite{vanwijk:2015}. But such `common distributed knowledge' modalities can be equally well interpreted on simplicial models:

Let $\C = (C,\chi,\ell)$ be a simplicial model. Define operation $*_{\mathcal B}$ for all $\mathcal B \subseteq \powerset(A)$ inductively as: for all $X,Y,Z \in \FF(C)$, $X \in *_{\mathcal B}(X)$, and if $Y \in *_{\mathcal B}(X)$ and there is $B \in \mathcal B$ with $B \subseteq \chi(Z \inter Y)$, then $Z \in *_{\mathcal B}(X)$. Then:

\[ \begin{array}{lcl}
\C,X \models \CD_{\mathcal B} \phi & \text{iff} & \C,Y \models \phi \text{ for all } Y \in \FF(C) \  \text{such that}\ Y \in *_{\mathcal B}(X)
\end{array} \]
Here, ${\mathcal B}$ can be any set of subsets of $A$. Of interest in combinatorial topology may be subsets of agents of the same dimension, i.e., all subsets of colours/agents of size $m < n$. Let us therefore define $*_m$ as a special case of the above $*_{\mathcal B}$ namely such that: $X \in *_m(X)$, and if $Y \in *_m(X)$ and there is $B \subseteq A$ with $|B|=m+1$ and $B \subseteq \chi(Z \inter Y)$, then $Z \in *_m(X)$.\footnote{The epistemic model equivalent would be: ${\sim^{*\inter}_m} := (\Union_{B \subseteq A}^{|B|=m+1} \sim^\inter_B)^*$.} We then get:
\[ \begin{array}{lcl}
\C,X \models \CD_m \phi & \text{iff} & \C,Y \models \phi \text{ for all } Y \in \FF(C) \  \text{such that}\ Y \in *_m(X)
\end{array} \]
Clearly, we have that $\CD_0 \phi \eq C_A \phi$. An interesting case to reason about simplicial models of dimension $n$ is $\CD_{n-1} \phi$. This describes truth in a manifold (see Section~\ref{section.topologicaltools}): if $\C,X \models \CD_{n-1} \phi$ and $\C$ is a manifold, then $\phi$ should be true on the entire complex. Otherwise, it is true on a certain subcomplex of $\C$, namely on the restriction of $\C = (C,\chi,\ell)$ to the set of vertices $v \in \VV(C)$ containing $X$ that is a manifold. Clearly, for every $m < n$, $\CD_m \phi$ describes truth on a similar subcomplex of $\C$ but where all facets intersect in dimension $m$.

The modalities $\CD_m \phi$ are a mere example. We can also consider modalities describing truth on the boundary of a manifold. We can consider variations of so-called {\em conditional common knowledge} \cite{jfaketal.lcc:2006} targetting features of simplicial complexes. And so on.
 
\begin{example}
Reconsider the simplicial model of Example \ref{example.local}, below on the left, call it $\C$, and consider as well a variant $\C'$ extending it, below on the right. We note that $\C$ is a manifold, whereas $\C'$ is not. However, it is tempting to see $\C'$ as two manifolds intersecting in the vertex named $b0$. Simplicial model $\C'$ represents agent $b$ being uncertain between two complexes with only difference the value of $c$, $p_c$ or $\neg p_c$. With the operator $\CD_1$ we can describe truth in the different regions of $\C'$. We recall that $\CD_0 = C_A$. Some typical truths in $\C$ and $\C'$ are:
\begin{itemize}
\item $\C,X \models C_{abc} p_c$: this is true as $p_c$ is true in facets $X,Y,Z$.
\item $\C',X \not\models C_{abc} p_c$: this is false as, for example, $\C',V \not\models p_c$, and facets $X$ and $V$ are connected via $c \in \chi(X \inter Z)$ and $b \in \chi(Z\inter V)$.
\item $\C',X \models \CD_1 p_c$: this is true, as we only need to consider connectivity by a sequence of facets intersecting in dimension $1$. We note that $|X\inter Y|=1$ and $|Y \inter Z|=1$, and that $p_c$ is true in $X, Y$ and $Z$. Therefore  $\C,X \models \CD_1 p_c$.
\item $\C',V \models \CD_1 \neg p_c$: in the other part of $\C'$, $CD_1$ describes truth in facets $V, U, W$.
\end{itemize}

\begin{center}
\scalebox{.8}{
$\C$: \quad
\begin{tikzpicture}
\node (a0) at (0,0) {$a0$};
\node (b0) at (4,0) {$b0$};
\node (a1) at (3,1.71) {$a1$};
\node (b1) at (1,1.71) {$b1$};
\node (c1) at (2,0) {$c1$};
\node (f1) at (1,.65) {$X$};
\node (f3) at (3,.65) {$Z$};
\node (f4) at (2,1.05) {$Y$};
\draw[-] (a0) -- (c1);
\draw[-] (a0) -- (b1);
\draw[-] (a1) -- (c1);
\draw[-] (a1) -- (b0);
\draw[-] (b0) -- (c1);
\draw[-] (b1) -- (c1);
\draw[-] (b1) -- (a1);
\end{tikzpicture}
\quad\quad
$\C'$: \quad
\begin{tikzpicture}
\node (a0) at (0,0) {$a0$};
\node (b0) at (4,0) {$b0$};
\node (a1) at (3,1.71) {$a1$};
\node (b1) at (1,1.71) {$b1$};
\node (c1) at (2,0) {$c1$};
\node (f1) at (1,.65) {$X$};
\node (f3) at (3,.65) {$Z$};
\node (f4) at (2,1.05) {$Y$};
\draw[-] (a0) -- (c1);
\draw[-] (a0) -- (b1);
\draw[-] (a1) -- (c1);
\draw[-] (a1) -- (b0);
\draw[-] (b0) -- (c1);
\draw[-] (b1) -- (c1);
\draw[-] (b1) -- (a1);
%
\node (ra0) at (8,0) {$a0$};
\node (rb1) at (7,1.71) {$b1$};
\node (ra1) at (5,1.71) {$a1$};
\node (rc1) at (6,0) {$c0$};
\node (rf1) at (5,.65) {$V$};
\node (rf3) at (7,.65) {$W$};
\node (rf4) at (6,1.05) {$U$};
\draw[-] (b0) -- (rc1);
\draw[-] (b0) -- (ra1);
\draw[-] (rb1) -- (rc1);
\draw[-] (rb1) -- (ra0);
\draw[-] (ra0) -- (rc1);
\draw[-] (ra1) -- (rc1);
\draw[-] (ra1) -- (rb1);
\end{tikzpicture}
}
\end{center}
\end{example}

\section{Belief for simplicial complexes} \label{section.belief}

Apart from knowledge, another common notion in epistemic logics is that of \emph{belief}. The received convention is that, unlike knowledge, beliefs may be false, but that otherwise belief has the properties of knowledge, including the property known as \emph{consistency}. The modality for belief is $B_a$. That beliefs may be false is embodied in the failure of the validity $B_a \phi \imp \phi$. Instead, we have the `consistency' axiom $B_a \phi \imp \hat{B}_a\phi$ (that corresponds to seriality of underlying frames). Like knowledge, belief should also satisfy the `introspection' principles $B_a \phi \imp B_a B_a \phi$ and $\neg B_a \phi \imp B_a \neg B_a \phi$.

More precisely, in the inductive language definition we can replace the inductive clause $K_a \phi$, for `agent $a$ knows $\phi$', with the inductive clause $B_a \phi$, for `agent $a$ believes $\phi$' (if we were to have both, a more complex proposal is required).

The validity $K_a \phi \imp \phi$ for knowledge characterizes the reflexivity of the epistemic models. With the weaker $B_a \phi \imp \hat{B}_a\phi$ for belief the relations in epistemic models are no longer equivalence relations. A typical model wherein agent $a$ incorrectly believes $p$ is $\model = (S,R,L)$ with $S = \{s,t\}$, $R = \{(s,t),(t,t)\}$ and $p \in L(t)$, often depicted as $\np \stackrel{a}{\imp} p$ (this visualization uses the convention for belief that if an arrow points to a state, that state implicitly has a reflexive arrow). However, this visualization is not relevant for modelling belief on simplicial models. 

See \cite{hintikka:1962,meyeretal:1995,hvdetal.handbook:2015} for more information on belief.

\paragraph*{Semantics of belief}
We now address the semantics of belief on simplicial models. Let a simplicial model $\C = (C, \chi,\ell)$ be given. For each agent $a \in A$, let $f_a$ be an idempotent function between the $a$-coloured vertices of $C$, called {\em belief function} (a function $f$ is idempotent iff for all $x$, $f(f(x)) = f(x)$). We emphasize that $f_a$ need not preserve variables: it may well be that $p_a \in \ell(v)$ but $p_a \notin \ell(f_a(v))$, or vice versa. 

We recall the semantics of knowledge on simplicial models:
\[ \begin{array}{lcl}
\C,X \models K_a \phi & \text{iff} & \C,X' \models \phi \ \text{for all} \ X' \ \text{with} \ a \in \chi(X\inter X')
\end{array} \]
Somewhat differently but equivalently formulated this is (we recall that $X_a$ is the vertex of facet $X$ coloured $a$):
\[ \begin{array}{lcl}
\C,X \models K_a \phi & \text{iff} & \C,X' \models \phi \ \text{for all} \ X' \ \text{with} \ X_a \in X'
\end{array} \]
The proposed semantics of belief is now as follows. It is surprisingly straightforward.
\[ \begin{array}{lcl}
\C,X \models B_a \phi & \text{iff} & \C,X' \models \phi \ \text{for all} \ X' \ \text{with} \ f_a(X_a) \in X'
\end{array} \]

\begin{proposition}[Properties of belief]
The following are validities of the logic of belief on simplicial complexes.
\begin{enumerate}
\item $B_a \phi \imp \hat{B}_a \phi$ (consistency)
\item $B_a \phi \imp B_a B_a \phi$ (positive introspection)
\item $\hat{B}_a \phi \imp B_a \hat{B}_a \phi$ (negative introspection)
\end{enumerate}
\end{proposition}
\begin{proof}
These properties are easily derived. The requirement that for all $a \in A$ the belief function $f_a$ is idempotent guarantees both introspection properties of belief.
\begin{itemize}
\item Belief is consistent, because $f_a$ is a (total) function and the simplicial complex is pure, so that for any $a$-coloured vertex $v$, $f_a(v)$ must be contained in at least one facet.
\item To show positive introspection we use the contrapositive $\hat{B}_a\hat{B}_a \neg\phi \imp \hat{B}_a \neg\phi$. Let $\C,X \models \hat{B}_a\hat{B}_a \neg\phi$, then there is a facet $X'$ such that $f_a(X_a) \in X'$ and $\C,X' \models \hat{B}_a \neg\phi$. Therefore, there are $X',X''$ such that: $f_a(X_a) \in X'$, $f_a(X'_a)\in X''$, and $\C,X'' \models \neg\phi$. As $f_a$ is idempotent and as $X'$ obviously only contains a single vertex of colour $a$, $f_a(X'_a) = f_a(f_a(X_a)) = f_a(X_a)$. Therefore, by the semantics of belief, $\C,X \models \hat{B}_a \neg\phi$.
\item Negative introspection is shown very similarly. Let $\C,X \models \hat{B}_a\phi$ be given, i.e., there is a facet $X'$ such that $f_a(X_a) \in X'$ and $\C,X' \models \phi$. We need to show that $\C,X \models B_a \hat{B}_a\phi$. In order to show that, let $X''$ be such that $f_a(X_a) \in X''$. Again, as $f_a(f_a(X_a))=f_a(X_a)$ we see that $\C,X'' \models \hat{B}_a\phi$, thus $\C,X \models B_a \hat{B}_a\phi$. 
\end{itemize} \vspace{-.5cm}
\end{proof}
A special case of belief is when agents correctly believe (and thus know) the value of their local state but may only be incorrect about the beliefs or the variables of other agents. This is guaranteed by requiring that for all $p_a \in P_a$, and all vertices $v$, $p_a \in \ell(v)$ iff $p_a \in \ell(f_a(v))$. We call such belief functions {\em locally correct}. It will be obvious that on a simplicial model where belief functions are locally correct, $B_a p_a \imp p_a$ and $B_a \neg p_a \imp \neg p_a$ are valid.

\paragraph*{Simplicial belief on epistemic models?}
The belief notion that we proposed on simplicial models results in rather particular  models in the Kripke model semantics. Take an epistemic model $(S,\sim,L)$ and let for each agent $a \in A$ belief function $f_a$ be an idempotent function between the equivalence classes $[s]_a$ of $a$. Define relation $R_a \subseteq S \times S$ as: $R_ass'$ iff there is a $t \in f([s]_a)$ with $t \sim_a s'$. Then $(S,R,L)$ is a model wherein all $R_a$ are serial, transitive, and Euclidean and thus interpret standard consistent belief  (a.k.a.\ ${\mathcal KD}45$ belief \cite{hvdetal.handbook:2015}). In other words, the agent's beliefs are uniform: the agent still believes to know her local state, but in all indistinguishable global states she considers it possible that the global states are those of another equivalence class, wherein her local state might be different.
\begin{example} \label{example.bigtri2}
Consider again the simplicial model from Example \ref{example.ins} depicted below. \begin{center}
\scalebox{.8}{
\begin{tikzpicture}
\node (a0) at (0,0) {$a0$};
\node (b0) at (4,0) {$b0$};
\node (c0) at (2,3.42) {$c0$};
\node (a1) at (3,1.71) {$a1$};
\node (b1) at (1,1.71) {$b1$};
\node (c1) at (2,0) {$c1$};
\node (f1) at (1,.65) {$F_1$};
\node (f2) at (2,2.36) {$F_2$};
\node (f3) at (3,.65) {$F_3$};
\node (f4) at (2,1.05) {$F_4$};
\draw[-] (a0) -- (c1);
\draw[-] (a0) -- (b1);
\draw[-] (a1) -- (c0);
\draw[-] (a1) -- (c1);
\draw[-] (a1) -- (b0);
\draw[-] (b0) -- (c1);
\draw[-] (b1) -- (c1);
\draw[-] (b1) -- (c0);
\draw[-] (b1) -- (a1);
\end{tikzpicture}
}
\end{center}
We give three examples of idempotent simplicial functions, and the resulting consequences for the beliefs of the agents:
\begin{enumerate}
\item $f_a$ defined as $f_a(a0) = a1$ (and thus $f_a(a1) = a1$).
\item $f_c$ defined as $f_c(c0) = c1$ (and thus $f_c(c1) = c1$).
\item $f'_a$ defined as $f'_a(a1) = a0$ (and thus $f'_a(a0) = a0$), $f'_b(b1) = b0$, and $f'_c(c1) = c0$.
\end{enumerate}
We now can, for example, validate the following statements:
\begin{enumerate}
\item $\C,F_1 \models \neg p_a \et B_a p_a$ \begin{quote} Agent $a$ mistakenly believes that its local state is $p_a$, whereas in fact it is $\neg p_a$. For the verification of the belief component we note that $\C,F_1 \models B_a p_a$ iff $\C,F' \models p_a$ for all $F'$ containing $f_a(a0) = a1$, which are $F_2,F_3,F_4$ that all obviously satisfy $a1$.\end{quote}
$\C,F_1 \models (p_b\et p_c) \et \hat{B}_a (p_b\et p_c)\et \neg B_a (p_b\et p_c)$ \begin{quote} The first conjunct is obvious, the second is because $a1 \in F_4$, the last is because $a1 \in F_2$ and also because $a1 \in F_3$, neither satisfy $p_b\et p_c$. \end{quote}
\item Similarly to the first item, we now have for example that $\C,F_2 \models \neg p_c \et B_c p_c$.
\item $\C, F_1 \models B_a (\neg p_a \et p_b \et p_c) \et B_b (p_a \et \neg p_b \et p_c) \et B_c (p_a \et p_b \et \neg p_c)$ \begin{quote} In $F_1$, the beliefs of the agents are inconsistent, they all think to know what the actual state of affairs (the designated facet) is, but their beliefs are incompatible. In fact, for any facet, under this $f'$ map, their beliefs about the valuation are mistaken. Also, in any facet, all three agents believe that they are certain (all agents only consider one facet possible). \end{quote}
\end{enumerate}
In this example no belief function is locally correct. If we we were to change the value of $p_a$ in vertex $a1$ into false, then $f_a$ and $f'_a$ would be locally correct.
\end{example}

\section{Simplicial action models} \label{section.actionmodels}

Action model logic is an extension of epistemic logic with operators modelling change of information \cite{baltagetal:1998,hvdetal.del:2007}. Although we can, of course, strictly separate syntax and semantics, it is common to see these operators as epistemic models wherein the valuations of atoms have been replaced by formulas (in the logical language), that should be seen as executability preconditions for the states in the epistemic model, that are therefore not called states but actions. Instead of presenting action models (for which we refer to the above references) we will present their simplicial equivalents, the \emph{simplicial action models} introduced in \cite{goubaultetal:2018}. We present them with a minor variation, namely with preconditions for vertices instead of facets, and also including so-called {\em factual change}. We recall that we assume that all simplicial complexes are chromatic pure simplicial complexes. This is also the case for simplicial action models.

\paragraph*{Simplicial action model}
A \emph{simplicial action model} $\C$ is a quadruple $(C,\chi,\pre,\post)$ where $C$ is a simplicial complex, $\chi$ a chromatic map, $\pre: \VV(C) \imp \lang_K(A,P)$ a precondition function assigning to each vertex $v \in \VV(C)$ a formula, and $\post: \VV(C) \imp P \imp \lang_K(A,P)$ a postcondition function assigning to each vertex and to each local variable for the colour of that vertex a formula.\footnote{In logics that contains modalities for such action models, we need to require that $\VV(C)$ is finite and that $\post$ is a partial function defined for a finite subset $P' \subseteq P$ only.}

If $\chi(v)=a$, $\pre(v) = \phi$, and $\post(v)(p) = \psi$ this should be seen as agent $a$ in $v$ executing the program ``If $\phi$ then $p := \psi$.'' Such assignments are simultaneous for all local variables.

\paragraph{Execution of a simplicial action model}
Let a simplicial model $\C = (C,\chi,\ell)$ and a simplicial action model $\C' = (C',\chi',\pre',\post')$ be given. The \emph{restricted product} of $\C$ and $\C'$ is the simplicial model $\C \otimes \C' = (C'',\chi'',\ell'')$ where: \begin{itemize} \item $C''$ is the modal product (Cartesian product) $C \times C'$ restricted to all facets $X \times X'$ such that $\C,X \models \Et_{v' \in X'} \pre'(v')$;\footnote{Using local semantics, instead of $\C,X \models \Et_{v' \in X'} \pre'(v')$ we can require that $\C, v \models \pre'(v')$ for all $v \in X$ and $v' \in X'$ with $\chi(v) = \chi(v')$. This may be more elegant.}
\item for all $(v,v')$ in the domain of $C''$, $\chi''((v,v')) = \chi(v)$ ($= \chi'(v')$);
\item for all $(v,v')$ in the domain of $C''$ and for all $p \in P$, $p \in \ell''((v,v'))$ iff $p \in \ell(v)$.
\end{itemize}

In \cite{goubaultetal:2018}, simplicial action models have preconditions associated to facets. The definitions are interchangeable: in one direction, for any facet $X$, $\pre(X) := \Et_{v \in X} \pre(v)$, as above, and in the other direction, for any vertex $v$, $\pre(v)  := \Vel_{v \in X} \pre(X)$. Defining preconditions on vertices has however some modelling advantages.

In the first place we can require, as a special case, all preconditions to be {\em local}, i.e., any $\pre(v)$ for vertex $v$ of colour $a$ has shape $K_a \phi$ (or, alternatively, any $\pre(v)$ is equivalent to $K_a \pre(v)$ on the simplicial model where the simplicial action model is executed). This would be in accordance with distributed computing methodology, where all actions are sending and receiving actions between agents. It rules out those enacted by the environment feeding new information to the system. Most but not all simplicial action models in \cite{goubaultetal:2018,ledent:2019} are local in this sense. An exception are  those for binary consensus \cite[Ex.\ 3.48]{ledent:2019}. One could even enforce commonly known preconditions, as in \cite{HalpernM17}.

In the second place we can require, as an even more special case, all preconditions to be in the language restricted to the agent colouring that vertex, i..e, any $\pre(v)$ for vertex $v$ of colour $a$ is a formula $\phi \in \lang_K(\{a\},P_a)$. As the value of all atoms $p_a$ is known by $a$, and given the properties of $K_a$, this implies that preconditions are local. Again, most simplicial action models in \cite{goubaultetal:2018,ledent:2019} are even local in that more restricted sense (e.g., for the immediate snapshot). 

A more significant difference from \cite{goubaultetal:2018} is that we also incorporate \emph{factual change} (change of the values of local variables) in our simplicial actions models, just as in dynamic epistemic logic \cite{jfaketal.lcc:2006,hvdetal.world:2008}. We may need factual change to describe snapshot algorithms with write actions, or binary consensus, of which we will give an example.

\begin{example} \label{example.bigtri3}
Once more we consider the simplicial model $\C$ from Example \ref{example.ins}, in which we execute the action representing that $a$ and $b$ do not know the value of $c$. We can think of this as a simplicial action model $\C'$ consisting of two facets intersecting in the edge $\{b,a\}$, with preconditions: $\pre(a') = \neg (K_a p_c \vel K_a \neg p_c)$, $\pre(b') = \neg (K_b p_c \vel K_b \neg p_c)$, $\pre(c'_0) = \neg p_c$, $\pre(c'_1) = p_c$. Below we depict that model and its execution.
\begin{center}
\scalebox{.8}{
\begin{tikzpicture}
\node (a0) at (0,0) {$a0$};
\node (b0) at (4,0) {$b0$};
\node (c0) at (2,3.42) {$c0$};
\node (a1) at (3,1.71) {$a1$};
\node (b1) at (1,1.71) {$b1$};
\node (c1) at (2,0) {$c1$};
\node (f1) at (1,.65) {$F_1$};
\node (f2) at (2,2.36) {$F_2$};
\node (f3) at (3,.65) {$F_3$};
\node (f4) at (2,1.05) {$F_4$};
\node (otimes) at (4.5,1.71) {\LARGE $\otimes$};
\draw[-] (a0) -- (c1);
\draw[-] (a0) -- (b1);
\draw[-] (a1) -- (c0);
\draw[-] (a1) -- (c1);
\draw[-] (a1) -- (b0);
\draw[-] (b0) -- (c1);
\draw[-] (b1) -- (c1);
\draw[-] (b1) -- (c0);
\draw[-] (b1) -- (a1);
\end{tikzpicture}
\quad
\begin{tikzpicture}
\node (c0) at (2,3.42) {$c'_0$};
\node (a1) at (3,1.71) {$a'$};
\node (b1) at (1,1.71) {$b'$};
\node (c1) at (2,0) {$c'_1$};
\node (f2) at (2,2.36) {$F'_1$};
\node (f4) at (2,1.05) {$F'_2$};
\node (is) at (4.5,1.71) {\LARGE $=$};
\draw[-] (a1) -- (c0);
\draw[-] (a1) -- (c1);
\draw[-] (b1) -- (c1);
\draw[-] (b1) -- (c0);
\draw[-] (b1) -- (a1);
\end{tikzpicture}
\quad
\begin{tikzpicture}
\node (c0) at (2,3.42) {$(c0,c'_0)$};
\node (a1) at (3,1.71) {$(a1,a')$};
\node (b1) at (1,1.71) {$(b1,b')$};
\node (c1) at (2,0) {$(c1,c'_1)$};
\draw[-] (a1) -- (c0);
\draw[-] (a1) -- (c1);
\draw[-] (b1) -- (c1);
\draw[-] (b1) -- (c0);
\draw[-] (b1) -- (a1);
\end{tikzpicture}
}
\end{center}
We now have, for example, that after this action $b$ still does not $c$'s value, but knows that $c$ now always knows his value: $C \otimes C', (F_4,F'_2) \models \neg (K_b p_c \vel K_b \neg p_c) \et K_b (K_c p_b \vel K_c \neg p_b)$.\footnote{If we were to incorporate modalities for simplicial action models into the logical language, we can also express propositions such as `before the update $c$ did not know the value of $b$'s variable, but afterwards she knows': $C, F_4 \models \neg (K_c p_b \vel K_c \neg p_b) \et [C',F'_2] (K_c p_b \vel K_c \neg p_b)$.} We have to choose a perspective of a pair of facets in the updated model in order to be able to evaluate any formula. For example, $C \otimes C', (F_4,F'_1) \models \neg p_c$ whereas $C \otimes C', (F_4,F'_2) \models p_c$. 
\end{example}

\begin{example}
As an example of factual change, consider the following action wherein agent $c$ publicly resets the value of her local variable to true, thus removing the initial uncertainty for $a$ and $b$ about that value. The latter is depicted in $\C$ below. 

The simplicial action model $C'$ consists of a single facet, also depicted, with vertices $a,b,c$, as usual reusing colours as names, with $\pre(a) = \pre(b) = \pre(c) = \top$ and $\post(c)(p_c) = \top$. We may assume the postconditions for the other variables to be trivial, i.e., $\post(a)(p_a) = p_a$ and $\post(c)(p_b) = p_b$. The action $c$ can be executed in both vertices $c0$ and $c1$, where we note that in the resulting simplicial model $\C \otimes \C'$ the valuation is such that $p_c \in \ell'(c0,c)$ and $p_c \in \ell'(c1,c)$. Therefore, the two-faceted updated simplicial model is bisimilar to the single-faceted one depicted to the right of it.

Note that a public \emph{assignment} of $p_c$ to $\top$ is different from a public \emph{announcement} that $p_c$ (is true). This would result in the same updated simplicial model restriction, but as a consequence of a restriction to the facet wherein $p_c$ is true.

\begin{center}
\scalebox{.8}{
\begin{tikzpicture}
\node (c0) at (2,3.42) {$c0$};
\node (a1) at (3,1.71) {$a1$};
\node (b1) at (1,1.71) {$b1$};
\node (c1) at (2,0) {$c1$};
\node (otimes) at (4,1.71) {\LARGE $\otimes$};
\draw[-] (a1) -- (c0);
\draw[-] (a1) -- (c1);
\draw[-] (b1) -- (c1);
\draw[-] (b1) -- (c0);
\draw[-] (b1) -- (a1);
\end{tikzpicture}
\quad
\begin{tikzpicture}
\node (c0) at (2,3.42) {$c$};
\node (a1) at (3,1.71) {$a$};
\node (b1) at (1,1.71) {$b$};
\node (c1) at (2,0) {};
\node (is) at (4,1.71) {\LARGE $=$};
\draw[-] (a1) -- (c0);
\draw[-] (b1) -- (c0);
\draw[-] (b1) -- (a1);
\end{tikzpicture}
\quad
\begin{tikzpicture}
\node (c0) at (2,3.42) {$(c0,c)$};
\node (a1) at (3,1.71) {$(a1,a)$};
\node (b1) at (1,1.71) {$(b1,b)$};
\node (c1) at (2,0) {$(c1,c)$};
\node (is) at (4.2,1.51) {\Large $\bisim$};
\draw[-] (a1) -- (c0);
\draw[-] (a1) -- (c1);
\draw[-] (b1) -- (c1);
\draw[-] (b1) -- (c0);
\draw[-] (b1) -- (a1);
\end{tikzpicture}
\quad
\begin{tikzpicture}
\node (c0) at (2,3.42) {$c1$};
\node (a1) at (3,1.71) {$a1$};
\node (b1) at (1,1.71) {$b1$};
\node (c1) at (2,0) {};
\draw[-] (a1) -- (c0);
\draw[-] (b1) -- (c0);
\draw[-] (b1) -- (a1);
\end{tikzpicture}
}
\end{center}
\end{example}


\begin{example}[Simplicial action model for binary consensus] \label{example:bincons}
In binary consensus, agents wish to agree on the value of a binary variable. We may assume that this variable is an atomic proposition, and that each agent has a copy that is its local variable, that they can communicate or manipulate while aiming to obtain consensus. Each agent $a$ thus has a variable $1_a$, where $\neg 1_a$ is $0_a$, in order to reach consensus on $0$ or $1$. We assume that all combinations of values are possible, but also that any kind of knowledge about the values of other agents is possible, thus abstracting from the part of the consensus algorithm obtaining such knowledge from other agents (e.g., possible knowledge gain modelled as subdivisions). This has the advantage that we can focus on the simplicial action model only.

Let there be two agents $a$, $b$ with variables $1_a$, $1_b$. For both agents there are four different actions, corresponding to four vertices. Consider $a$. If $K_a (1_a\et 1_b)$, then she leaves the value unchanged  (aiming for consensus on value $1$), if $K_a (0_a \et 0_b)$, she also leaves the value unchanged (aiming for consensus on value $0$). Otherwise, she considers it possible that her value and that of $b$ might be different. We note that $\neg K_a (1_a\et 1_b)\et\neg K_a (0_a \et 0_b)$ implies $\M_a (1_a \vel 1_b) \et \M_a (0_a \vel 0_b)$. In that case she non-deterministically resets her value to $0$ or to $1$. This is an assignment $1_a := \top$ or $1_a := \bot$ (i.e., reset to value $0$). Agent $a$ can distinguish all these actions, whereas agent $b$ cannot distinguish any of those. The four actions of agent $b$ are analogous, and cannot be distinguished by $a$. 

The simplicial action model is therefore defined as $\C = (C,\chi,\pre,\post)$ where $\VV(C) = \{u_i \mid i = a,b \ \text{and}\ u = w,x,y,z\}$ and $\FF(C) = \{(u_a,v_b) \mid u,v = w,x,y,z\}$, and $\chi(u_i) =  i$ for $i = a,b$ and $u = w,x,y,z$. Then, for $i = a,b$:
\[ \begin{array}{lllll}
&& \pre(w_i) & = & K_i (1_a\et 1_b) \\
\pre(x_i) &=& \pre(y_a) & = & \neg K_i (1_a\et 1_b) \et \neg K_i (0_a \et 0_b) \\
&& \pre(z_i) & = & K_i (0_a\et 0_b)
\end{array}\]
and
\[ \begin{array}{lllll}
\post(w_i)(1_i) & = & \post(z_i)(1_i) & = & 1_i \\
&& \post(x_i)(1_i) & = & \top \\
&& \post(y_i)(1_i) & = & \bot \\
\end{array}\]
The assignment $\post(z_i)(1_i)= 1_i$ means that the value of $1_i$ remains the same: if agent $i$ has a $1$, it remains $1$, and if she has a $0$, it remains $0$. 

As is to be expected, unless $a$ and $b$ have common knowledge that the value is $1$ or common knowledge that the value is $0$, they cannot obtain such common knowledge by executing this simplicial action $\C$: binary consensus is impossible. Even when they both know that the values agree, for example when $K_a (1_a\et 1_b)$ and $K_b (1_a\et 1_b)$, they may still consider it possible that the other does not know that, as when $\M_a \neg K_b (1_a\et 1_b)$, in which case $a$ considers it possible that $b$ has reset its value to $0$, such that consensus would be lost. Although in fact, because $K_a (1_a\et 1_b)$ and $K_b (1_a\et 1_b)$, they both kept their $1$s: there is consensus. And so on, forever.

Validity and agreement are therefore not guaranteed by this simplicial action model. Validity can be obtained for slight variations, for example when agents do not randomly choose a value in case they consider it possible other agents have other values, as here, but choose the value of the majority. Under that policy, then if all agents know all values and there is a majority, validity is guaranteed.

Generalizing the case from $2$ agents to $n$ agents is straightforward. Instead of $2 \cdot 4 = 8$ vertices we get $4n$ vertices: {\em linear} growth. Interestingly, the corresponding action model grows {\em exponentially}, as each action in the action model has a precondition that is a conjunction of $n$ parts, where each conjunct is one of these $4$ options: $4^n$ actions. The distributed modelling is decidedly more elegant.
\end{example}

\section{Conclusion and further research} \label{section.conclusion}

We have presented a survey of epistemic tools and techniques that can be used to model knowledge and change of knowledge on simplicial complexes, which are a representation of high-dimensional discrete topological spaces. Many of these notions bring in novel issues in this topological setting. In future research we wish to apply these tools to model typical distributed computing tasks and algorithms, beyond the situations studied in~\cite{goubaultetal:2019,ledent:2019,goubaultetal:2018}, and to explore further epistemic horizons for these and other applications.

An interesting possibility of interaction emerges with algebraic topology, a very deep and mature area of mathematics. For example, the well-known notion of \emph{consensus} depends on the connectivity of an epistemic frame and the related notion of common knowledge~\cite{halpernmoses:1990}; this is a $1$-dimensional topological property. However, other notions of knowledge appear to be related to higher-dimensional topological properties, and to forms of agreement weaker than consensus~\cite{GoubaultLLR19}.

With the exception of Example~\ref{ex:fig-synch} modelling synchronous computation, we did not investigate impure complexes. Such complexes may well be interpreted in terms of unawareness of local variables as in~\cite{faginetal:1988} (and as unawareness of agents). Special notions of bisimulation would fit impure complexes.

As suggested in \cite{ledent:2019}, distributed knowledge may be employed to describe {\em simplicial sets}, the generalization of simplicial complexes to multi-sets (the \emph{pseudocomplexes} of~\cite{hilton_wylie_1960}, with a similar duality to graphs~\cite{BrachoM1987}). Then, we could have two facets that are indistinguishable to \emph{all} agents, aiming at removing the requirement that the epistemic model be proper. 

It would also be interesting to consider infinite models. In distributed computing, infinitely many states and infinitely many agents (processes) have been considered in e.g.~\cite{AguileraInf04}. In modal logic, many tools and techniques apply just as well to the countably infinite situation.

\bibliographystyle{plain}
\bibliography{biblio2019}

\end{document}